\documentclass[a4paper,leqno]{amsart}

\usepackage{latexsym}
\usepackage[english]{babel}
\usepackage[mathscr]{eucal}
\usepackage{amsmath}
\usepackage{amsthm}
\usepackage{amsfonts}
\usepackage{amssymb}
\usepackage{amscd}
\usepackage{bbm}
\usepackage{graphicx}
\usepackage{graphics}
\usepackage{latexsym}
\usepackage{color}

\usepackage{harvard}
\usepackage{cite}
\usepackage{graphicx}
\usepackage{graphics}

\newcommand{\ud}{\mathrm{d}}

\newcommand{\ii}{\mathrm{i}}
\newcommand{\cH}{\mathcal{H}}

\theoremstyle{plain}
\newtheorem{theorem}{Theorem}[section]
\newtheorem{lemma}[theorem]{Lemma}

\theoremstyle{definition}

\newtheorem{remark}[theorem]{Remark}
\newtheorem*{remark*}{Remark}

\numberwithin{equation}{section}

\begin{document}

\title[Effective spinor dynamics in a spin-1 Bose condensate]{Effective non-linear spinor dynamics \\ in a spin-1 Bose-Einstein condensate}

\author[A.~Michelangeli]{Alessandro Michelangeli}
\address[A.~Michelangeli]{International School for Advanced Studies -- SISSA \\ via Bonomea 265 \\ 34136 Trieste (Italy).}
\email{alessandro.michelangeli@sissa.it}
\author[A.~Olgiati]{Alessandro Olgiati}
\address[A.~Olgiati]{International School for Advanced Studies -- SISSA \\ via Bonomea 265 \\ 34136 Trieste (Italy).}
\email{alessandro.olgiati@sissa.it}
	
	\begin{abstract}
	We derive from first principles the experimentally observed effective dynamics of a spinor Bose gas initially prepared as a Bose-Einstein condensate and then left free to expand ballistically. In spinor condensates, which represent one of the recent frontiers in the manipulation of ultra-cold atoms, particles interact with a two-body spatial interaction and a spin-spin interaction. The effective dynamics is well-known to be governed by a system of coupled semi-linear Schr\"{o}dinger equations: we recover this system, in the sense of marginals in the limit of infinitely many particles, with a mean-field re-scaling of the many-body Hamiltonian. When the resulting control of the dynamical persistence of condensation is quantified with the parameters of modern observations, we obtain a bound that remains quite accurate for the whole typical duration of the experiment.
	\end{abstract}

\vspace{2pc}

%\submitto{\jpa}
%\maketitle
\maketitle
\section{Spinor condensation and emergent effective dynamics: \\ setting of the problem}\label{sec:intro}

The study of superfluid systems with internal degrees of freedom has been a tantalising goal of cold atom physics for long. The first experimental breakthroughs were found with magnetically-trapped gases of $^{87}\mathrm{Rb}$ that turned out to be quite long lived due to the fortunate circumstance that the singlet and triplet scattering lengths have almost the same value, which greatly suppresses the spin-spin collision rate. This allows for the creation of Bose-Einstein condensates where a macroscopic occupation of particles can be driven coherently through two hyperfine levels, typically the $|F=1,m_F=-1\rangle$, the $|F=2,m_F=2\rangle$, or the $|F=2,m_F=1\rangle$ states. Such systems are customarily referred to as \emph{quasi-spinor condensates}: only a restricted number of hyperfine levels contribute effectively to the experiment, through the coupling with a resonant external magnetic field (Rabi coupling).

In contrast, a highly off-resonant magnetic confinement can trap the atoms irrespectively of their hyperfine state: in this case the spin becomes a new degree of freedom and this produces interacting Bose gases of ultra-cold atoms where the spatial two-body interaction is mediated by a spin-spin coupling. When such systems exhibit a macroscopic occupation of the same one-body state, the order parameter being now a vector in the hyperfine spin space, one speaks of \emph{spinor condensates}. For them, the hallmark of condensation manifests as a reversible spin-changing collisional coherence between particles.

For $s$-wave interactions, the rotational symmetry of the two-body collisions among atoms of hyperfine spin $F_1$ and $F_2$ can only depend on their total spin $F_\mathrm{tot}:=F_1+F_2$ and not on its orientation: for identical interacting bosons of spin $F$, $F_\mathrm{tot}$ can take values in $\{0,2,\dots,2F\}$, thus reducing the inter-particle short-range interaction to the form
\[
\sum_{\mathrm{f}\in\{0,2,\dots,2F\}}a_{\mathrm{f}}\,P_{\mathrm{f}}\,,
\]
whenever any two particles collide, where $a_{\mathrm{f}}$ is the $s$-wave scattering length for collisions between particles with total spin $\mathrm{f}$ and $P_{\mathrm{f}}$ is the projection onto the space of total spin $\mathrm{f}$.

In the concrete case of spin-1 bosons \cite{Ho-1998} the total spin in any two-particle channel can be 0 or 2 and the above expression for the interaction can be re-written as
\begin{equation}\label{eq:spin_interaction}
c_0+c_2\,\mathbf{F_1}\cdot\mathbf{F_2}\,,
\end{equation}
where $\mathbf{F_j}=(F_j^{(1)},F_j^{(2)},F_j^{(3)})$ is the spin operator of the $j$-th particle, $j\in\{1,2\}$, and
\begin{equation}\label{eq:scatt_lenghts_comb}
c_0\;=\;{\textstyle\frac{1}{3}}(a_0+2a_2)\,,\qquad  c_2\;=\;{\textstyle\frac{1}{3}}(a_2-a_0)\,.
\end{equation}

%being $c_0=(a_0+2a_2)/3$ and $c_2=(a_2-a_0)/3$.

The earliest theoretical investigations and observations of spinor condensates appeared some 20 years ago \cite{Ohmi:Machida:1998,Ho-1998,Law-Pu-Bigelow-1998}. For a modern experiment with $F=1$ $^{87}\mathrm{Rb}$ we refer to \cite{Chang-Qin-Zhang-You-2005}. By now the field has expanded through a vast series of experimental and theoretical studies, for a survey of which we refer to the comprehensive reviews \cite{Ketterle_StamperKurn_SpinorBEC_LesHouches2001, Malomed2008_multicompBECtheory, Hall2008_multicompBEC_experiments,StamperKurn-Ueda_SpinorBose_Gases_2012}. 
Such studies covered the ground state structure and the coherent spinor dynamics, among many other topics.

In order to provide firm grounds to the formal treatments available in the physical literature, in this work we present a rigorous derivation, from the `first principles' many-body linear Schr\"{o}dinger dynamics, of the well-known system of coupled non-linear equations that govern the evolution of the spinor condensate. This also complements our previous analysis of the emergent non-linear dynamics of pseudo-spinor condensates \cite{MO-pseudospinors-2017} and continues the programme outlined in \cite{AO-GPmixture-2016volume}.

We focus for concreteness on one of the most commonly observed scenarios: spin-1 Bose-Einstein condensates (with a natural extension of our statements and proofs to species with higher spin). 
%In this case the admissible spaces of total spin correspond to $F_{\mathrm{tot}}=0$ and $F_{\mathrm{tot}}=2$, and the two-body interaction consists of a purely spatial contribution with scattering length $a_{\{F_{\mathrm{tot}}=0\}}=:a_0$ and a spin-spin contribution with scattering length $a_{\{F_{\mathrm{tot}}=2\}}=:a_2$. 
Let us then set up the model as follows.

The particle number throughout this work is expressed by $N$ for some $N\in\mathbb{N}$ with $N\geqslant 2$. We denote collectively by $\pmb{\sigma}$ the symbolic vector $\pmb{\sigma}=(\sigma^{(1)},\sigma^{(2)},\sigma^{(3)})$, where
\begin{equation*}
\sigma^{(1)}\,=\,\frac{1}{\sqrt{2}}
\begin{pmatrix}
0 & 1 & 0 \\
1 & 0 & 1 \\
0 & 1 & 0 \\\end{pmatrix},\;\;
\sigma^{(2)}\,=\,\frac{-\ii}{\sqrt{2}}
\begin{pmatrix}
0 & 1 & 0 \\
-1 & 0 & 1 \\
0 & -1 & 0 \\\end{pmatrix},\;\;
\sigma^{(3)}\,=\,
\begin{pmatrix}1 & 0 & 0 \\
0 & 0 & 0 \\
0 & 0 & -1 \\\end{pmatrix}
\end{equation*}
are the usual Pauli matrices, regarded also as operators on $\mathbb{C}^3$\,. With respect to the tensor product $(\mathbb{C}^3)^{\otimes N}$ the notation  $\pmb{\sigma}_j=(\sigma^{(1)}_j,\sigma^{(2)}_j,\sigma^{(3)}_j)$ is going to be used to indicate the operator that acts as $\pmb{\sigma}$ on the $j$-th copy of the tensor product (the $j$-th spin degree of freedom), and trivially as the identity on all other copies. Analogously the notation $\pmb{\sigma}_j\bullet\pmb{\sigma}_k$, for given $j,k\in\{1,\dots,N\}$, is a short-cut for the operator
\[
\pmb{\sigma}_j\bullet\pmb{\sigma}_k\;=\;\sigma^{(1)}_j\otimes\sigma^{(1)}_k+\sigma^{(2)}_j\otimes\sigma^{(2)}_k+\sigma^{(3)}_j\otimes\sigma^{(3)}_k\,,
\]
understanding $A_j\otimes B_k$ as an operator acting non-trivially only on the $j$-th and the $k$-th copy of the tensor product space. We shall also simply write $\pmb{\sigma}_j\bullet\pmb{\sigma}_k=\sum_{\ell=1}^3\sigma^{(\ell)}_j\sigma^{(\ell)}_k$, thus omitting the tensor product symbol.

For a generic element $\phi$ of the \emph{one-body} Hilbert space
\begin{equation}
\mathfrak{h}\;:=\;L^2(\mathbb{R}^3,\ud x)\otimes\mathbb{C}^3\;\cong\; L^2(\mathbb{R}^3,\mathbb{C}^3)\,,
\end{equation}
namely a \emph{spinor} function, and for $s\in\mathbb{R}$, we shall write
\[
\phi\,\equiv\,\begin{pmatrix} u \\ v \\ w \\\end{pmatrix}\,,\qquad \|\phi\|_{H^s}^2\;=\;\|u\|_{H^s(\mathbb{R}^3)}^2+\|v\|_{H^s(\mathbb{R}^3)}^2+\|w\|_{H^s(\mathbb{R}^3)}^2
\]
where $u,v,w\in L^2(\mathbb{R}^3)$ are the components of $\phi$. Analogously, an element $\Psi_N$ of the \emph{many-body} Hilbert space
\begin{equation}
\cH_N\;:=\;\big(L^2(\mathbb{R}^3,\ud x)\otimes\mathbb{C}^3\big)^{\otimes N}
\end{equation}
can be canonically represented as a function
\[
 (\mathbb{R}^3)^N\times\{1,2,3\}^{N}\;\ni\;(x_1,\dots,x_N)\times(s_1,\dots,s_N)\;\mapsto\;\Psi_N(x_1,s_1;\cdots,x_N,s_N)
\]
that undergoes the scalar product
\begin{equation}
\begin{split}&\langle\Psi_N,\Phi_N\rangle_{\cH_N}\;:=\;\!\sum_{\substack{s_1,\dots,s_N\\\in\{1,2,3\}}}\\
	&\times\int_{(\mathbb{R}^3)^N}\ud x_1\cdots\ud x_N\,\overline{\Psi_N(x_1,s_1;\cdots;x_N,s_N)}\,\Phi_N(x_1,s_1;\cdots;x_N,s_N)
	\end{split}
\end{equation}
and for which therefore $\|\Psi_N\|_{\cH_N}^2:=\langle\Psi_N,\Psi_N\rangle_{\cH_N}<+\infty$.

Pure states of many-body systems of identical spin-1 bosons are represented by normalised vectors belonging to the bosonic Hilbert subspace
\begin{equation}
\cH_{N,\mathrm{sym}}\;:=\;\big(L^2(\mathbb{R}^3,\ud x)\otimes\mathbb{C}^3\big)^{\otimes_{\mathrm{sym}}N}
\end{equation}
of $\cH_N$, namely by $\Psi_N\in\cH_N$ with $\|\Psi_N\|_{\cH_N}=1$ that are symmetric under any exchange $(x_j,s_j)\leftrightarrow(x_k,s_k)$.

Acting on the bosonic space $\cH_{N,\mathrm{sym}}$ we consider the $N$-body mean-field Hamiltonian
\begin{equation}\label{eq:meanfieldHN}
 H_N\;:=\;\sum_{j=1}^N(-\Delta_{x_j})+\frac{1}{N}\sum_{1<j\leqslant k<N}W(x_j-x_k)+\frac{1}{N}\sum_{1<j\leqslant k<N}V(x_j-x_k)\,\pmb{\sigma}_j\bullet\pmb{\sigma}_k
\end{equation}
for given measurable $W:\mathbb{R}^3\to\mathbb{R}$ and $V:\mathbb{R}^3\to\mathbb{R}$ such that $H_N$ is unambiguously realised as a self-adjoint operator. % on $\cH_N$, or simply on $\cH_{N,\mathrm{sym}}$. 

In fact, what is to be assigned are two functions $W$ and $V$ such that when $N$ equals the experimental number of particle the scattering length of the potentials $\frac{1}{N}W$ and $\frac{1}{N}V$ have precisely the values $c_0$ and $c_2$ from \eqref{eq:spin_interaction}-\eqref{eq:scatt_lenghts_comb}.

The mean-field scaling in \eqref{eq:meanfieldHN} is set so as to eventually study the limit $N\to +\infty$ and to let the asymptotic dynamics emerge: although this is only a surrogate of the thermodynamic limit, it retains an amount of physical realism in the case of Bose gases at high dilution and weak interactions, and moreover it guarantees that the dynamics generated by $H_N$ remains non-trivial at any $N$, for the formal contributions of the kinetic and of the potential parts in $H_N$ have the same order $O(N)$. This also allows us to present a clean discussion of the rigorous emergence of the effective spinor dynamics, even if the technique exploited here is completely applicable also to physically more realistic and hence mathematically involved scalings, as in fact we did in a recent work \cite{MO-pseudospinors-2017} for the related model of \emph{pseudo-spinor} Bose gases. We defer a discussion of the analogue of our results for more singular scalings to the end Section \ref{sec:results}.

The class of experiments we are concerned with correspond to studying the Cauchy problem for the associated (linear) Schr\"odinger equation
\begin{equation}\label{eq:Cauchy_problem}
\begin{cases}
\;\ii \partial_t\Psi_N(t)\;=\;H_N\Psi_N(t) \\
\;\Psi_N(0)\;=\;\Psi_{N,0}
\end{cases}
\end{equation}
for a given initial datum $\Psi_{N,0}\in \cH_{N,\mathrm{sym}}$ that exhibits complete BEC asymptotically in $N$ in the usual sense of the \emph{reduced density matrix}, namely when condensation is monitored with observables relative to finitely many (fixed in $N$) particles.

More precisely, to each $\Psi_{N}\in\cH_{N,\mathrm{sym}}$, or more generally to each many-body density matrix $\gamma_N$ on $\cH_{N,\mathrm{sym}}$,  one associates the so-called one-body marginal (or one-body reduced density matrix)
\begin{equation}\label{eq:partial_trace1}
\gamma_N^{(1)}\;=\;\mathrm{Tr}_{N-1}\,\gamma_N\,,
\end{equation}
where the map $\mathrm{Tr}_{N-1}:\mathcal{B}_1(\cH_{N,\mathrm{sym}})\to\mathcal{B}_1(\mathfrak{h})$ is the \emph{partial trace} from trace class operators acting on $\cH_{N,\mathrm{sym}}$ to trace class operators acting on $\mathfrak{h}$ (we refer, e.g., to  \cite[Section 1]{M-Olg-2016_2mixtureMF} or \cite[Section 1]{MO-pseudospinors-2017} for the mathematical details). In short, 
% $\mathrm{Tr}_{N-1}\,\gamma_N$ is defined, by duality, by
% \begin{equation}\label{eq:def_partial_trace_without_basis}
% \Tr_\mathfrak{h}(A\cdot\Tr_{N-1}\,\gamma_N)\;=\;\Tr_{\cH}(A\otimes\mathbbm{1}_{N-1})\cdot \gamma_N))\qquad\forall A\in\mathcal{B}(\mathfrak{h})
% \end{equation}
% (here $\mathcal{B}$ denotes the bounded linear operators on $\mathfrak{h}$ and $\mathcal{B}_1$ the corresponding trace class).
% In terms of an arbitrary orthonormal basis $(\Xi_k)_k$ of $\cH_{N-1,\mathrm{sym}}$ one then has
% \begin{equation}\label{eq:def_partial_trace_with_basis}
% \langle \phi,(\mathrm{Tr}_{N-1}\,\gamma_N)\psi\rangle_\mathfrak{h}\;=\;\sum_{k}\langle\phi\otimes\Xi_k,\gamma_N\,\psi\otimes\Xi_k\rangle_{\cH_{N-1,\mathrm{sym}}}\qquad\forall\phi,\psi\in\mathfrak{h}\,,
% \end{equation}
% and the l.h.s.~of \eqref{eq:def_partial_trace_with_basis} is independent of the choice of the basis.
% Thus, 
$\gamma_N^{(1)}$ is obtained by ``tracing out'' $N-1$ degrees of freedom from $\gamma_N$, and more general $k$-body reduced density matrices $\gamma^{(k)}_N$, for integer $k <N$, are defined in a completely analogous way.

As a non-negative, bounded, and self-adjoint operator on $\mathfrak{h}$,  $\gamma_N^{(1)}$ has a complete set of real non-negative eigenvalues that sum up to 1, and thanks to the bosonic symmetry each such eigenvalue has the natural interpretation of \emph{occupation number}. Thus, complete spinor BEC of the many-body state $\Psi_N$ onto the one-body orbital $\phi_0\in\mathfrak{h}$ consists by definition of the property
\begin{equation}\label{eq:initialgamma}
\begin{split}\lim_{N\to\infty}\gamma_{N}^{(1)}\;&=\;\left|\!\begin{pmatrix} u_0 \\ v_0 \\ w_0\\\end{pmatrix}\!\right\rangle\left\langle\!\begin{pmatrix}u_0 \\ v_0 \\ w_0\\\end{pmatrix}\right|\,,\qquad \phi_0=\begin{pmatrix}u_0 \\ v_0 \\ w_0\\\end{pmatrix}\,, \\
\|\phi_0\|_{\mathfrak{h}}\;=\;\|&u_0\|_{L^2(\mathbb{R}^3)}^2+\|v_0\|_{L^2(\mathbb{R}^3)}^2+\|w_0\|_{L^2(\mathbb{R}^3)}^2\;=\;1\,,\end{split}
\end{equation}
for the reduced density matrix $\gamma_N$ associated with $\Psi_N$. The bounds
\begin{equation}\label{eq:equivalent-BEC-control}
\begin{split}1-\langle\phi,\gamma_N^{(1)}\phi\rangle_{\!L^2(\mathbb{R}^3)\otimes\mathbb{C}^3}\;&\leqslant\;\mathrm{Tr}_{\!L^2(\mathbb{R}^3)\otimes\mathbb{C}^3}\big|\,\gamma_{N}^{(1)}-|\phi\rangle\langle\phi|\,\big| \\
&\leqslant\;2\sqrt{1-\langle\phi,\gamma_N^{(1)}\phi\rangle_{\!L^2(\mathbb{R}^3)\otimes\mathbb{C}^3}} \qquad \forall\phi\in\mathfrak{h}\end{split}
\end{equation}
(see, e.g., \cite[Eq.~(1.8)]{M-Olg-2016_2mixtureMF})
show that the limit \eqref{eq:initialgamma} holds equivalently in the trace norm topology or in the weak operator topology, as well as in all topologies in between. Moreover, \eqref{eq:initialgamma} is actually equivalent to the analogous limit of $k$-marginals, that is, $\lim_{N\to\infty}\gamma_{N}^{(k)}\to|\phi_0^{\otimes k}\rangle\langle\phi_0^{\otimes k}|$, for every fixed $k$ \cite{am_equivalentBEC}.

Besides being clearly satisfied in the ideal case of a completely factorised $\Psi_N=\phi_0^{\otimes N}$, condition \eqref{eq:initialgamma} holds for an ample class of many-body states where the exact factorisation is partially altered by a pattern of inter-particle correlations that are small in $N$ at the level of the reduced marginals. In fact, \eqref{eq:initialgamma} provides a control on $\Psi_N$ that is obviously much weaker than the asymptotic smallness of $\|\Psi_N-\phi_0^{\otimes N}\|_{\cH_{N}}$ and yet is physically meaningful, being based on the expectations of `real-world' one-body (or few-body) observables.

The above discussion provides the appropriate setting for a rigorous control of the dynamical problem \eqref{eq:Cauchy_problem} when at time $t=0$ the system is prepared in a many-body state $\Psi_{N,0}$ with complete spinor BEC: owing to the non-trivial potential term in $H_N$ the exact solution $\Psi_N(t)=e^{-\ii t H_N}\Psi_{N,0}$ is clearly out of reach, both analytically and numerically, but it is still relevant to qualify it in the sense of the one-body marginal $\gamma_{N}^{(1)}(t)$.

Indeed, it is an extensive experimental evidence that under suitable conditions, and with very good approximation, spinor BEC $\gamma_N^{(1)}(t)\approx|\phi(t)\rangle\langle\phi(t)|$ persists at later times onto some spinor $\phi(t)$ for the components of which it is possible to measure the spatial densities $|u(x,t)|^2$, $|v(x,t)|^2$, and $|w(x,t)|^2$ (see, e.g., \cite[Fig.~1]{Chang-Qin-Zhang-You-2005}).

What the expected equations are that govern the evolution of such densities is inferred in the physical literature by means of various heuristics, that eventually follow all from the same line of reasoning: to replace formally $W(x_1-x_2)$ and $V(x_1-x_2)$ respectively with $c_0\delta(x_1-x_2)$ and $c_2\delta(x_1-x_2)$ and to plug the formal solution $\Psi_N(t)=\phi(t)^{\otimes N}$ into the Schr\"{o}dinger equation \eqref{eq:Cauchy_problem}, more precisely into the emergent energy-per-particle functional $\frac{1}{N}\langle\Psi_N(t),H_N\Psi_N(t)\rangle_{\cH_N}$ when $N\to +\infty$, and to determine the associated `one-body' Euler-Lagrange equation. A somewhat extended discussion on such heuristics in the analogous context for condensate mixtures can be found \cite[Sec.~4]{M-Olg-2016_2mixtureMF}.

This procedure yields, in very good agreement with the experiments, a system of coupled non-linear cubic Schr\"{o}dinger equations of the type \eqref{eq:MF_system} and \eqref{eq:GP_system} considered in this work. We shell refer to them as the \emph{emergent effective dynamics for the spinor condensate}.

In this mean-field picture one therefore moves from the enormously complicated linear dynamics of a large number $N$ of particles to an effective description, reduced to only one orbital, at the price of an emergent non-linearity, which is the signature of the self-interaction of each particle of the condensate with the density of the others.

What we therefore present in this work is the rigorous derivation of such an effective dynamics, which amounts to closing the diagram
\begin{equation*}\label{scheme_for_marginals}
\begin{CD}
\Psi_N @>\scriptsize\textrm{partial trace}>>\gamma_N^{(1)} @>N\to\infty>> \left| \!\begin{pmatrix}u_0 \\ v_0 \\ w_0 \\\end{pmatrix}\right\rangle\left\langle\!\! \begin{pmatrix}u_0 \\ v_0 \\ w_0 \\\end{pmatrix}\right| \\
@ V\scriptsize\begin{array}{c} \textrm{many-body} \\ \textrm{\textbf{linear} dynamics}  \end{array} VV @V  VV               @VV\scriptsize\begin{array}{c} \\ \textrm{\textbf{non-linear}} \\ \textrm{Schr\"{o}dinger eq.} \\ {\,} \end{array}V    \\
\Psi_{\! N}(t) @>\scriptsize\textrm{\qquad\qquad\qquad\;}>>\gamma_{\! N}^{(1)}\!(t) @>N\to\infty>> \left| \!\begin{pmatrix}u(t) \\ v(t) \\ w(t)\\\end{pmatrix}\right\rangle\left\langle\!\! \begin{pmatrix}u(t) \\ v(t) \\ w(t) \\\end{pmatrix}\right|\,.
\end{CD}
\end{equation*}

\section{Main results}\label{sec:results}

Let us formulate in this Section our main assumptions and results, to which we will add a few explanatory remarks and further comments.

First, in order to monitor the displacement between the reduced density matrix $\gamma_{\Psi_{\!N}}^{(1)}$
%$\gamma_{\Psi_{\!N}}$ 
associated with a normalised vector $\Psi_{\!N}\in\cH_{N,\mathrm{sym}}$ and the rank-one projection $|\phi\rangle\langle\phi|$ onto a spinor $\phi\in\mathfrak{h}$, and in view of \eqref{eq:initialgamma}-\eqref{eq:equivalent-BEC-control} above, we associate to $\Psi_{\!N}$ and $\phi$ the quantity
\begin{equation}\label{eq:defalpha}
\alpha_{(\Psi_{\!N},\phi)}\;:=\;1-\langle \phi,\gamma_{\Psi_{\!N}}^{(1)}\phi\rangle_\mathfrak{h}\,.
\end{equation}

We shall shorten henceforth the notation of Section \ref{sec:intro} by using subscripts to denote the time dependence, thus writing $\Psi_{N,t}$, $u_t$, etc., instead of $\Psi_N(t)$, $u(t)$, etc.

Next, let us introduce the mean-field non-linear Schr\"{o}dinger spinor system, or \emph{spinor Hartree system} for short,
\begin{equation}\label{eq:MF_system}
\begin{split}\ii\partial_t u\;&=\;-\Delta u+ (W*(|u|^2+|v|^2+|w|^2))u \\
&\qquad\quad+(V*|u|^2)u-(V*|w^2|)u+(V*(\overline{v}u))v+(V*(\overline{w}v))v \\
\ii\partial_t v\;&=\;-\Delta u+ (W*(|u|^2+|v|^2+|w|^2))v \\
&\qquad\quad+(V*(\overline{u}v))u+(V*(\overline{v}w))u +(V*(\overline{v}u))w+(V*(\overline{w}v))w  \\
\ii\partial_t w\;&=\;-\Delta u+ (W*(|u|^2+|v|^2+|w|^2))w \\
&\qquad\quad+(V*|w|^2)w-(V*|u^2|)w+(V*(\overline{u}v))v+(V*(\overline{v}w))v\end{split}
\end{equation}
in the unknowns $u\equiv u_t(x)$, $v\equiv v_t(x)$, and $w\equiv w_t(x)$.

In the notation of Section \ref{sec:intro} the system \eqref{eq:MF_system} can be re-written as
\begin{equation}\label{eq:MF_system_compact}
\begin{split}
\ii\partial_t\!\begin{pmatrix}u \\ v \\ w\\\end{pmatrix}&=\;-\Delta\!\begin{pmatrix}u \\ v \\ w\end{pmatrix}+\begin{pmatrix}(W*(|u|^2+|v|^2+|w|^2))\,u \\ (W*(|u|^2+|v|^2+|w|^2))\,v \\ (W*(|u|^2+|v|^2+|w|^2))\,w\end{pmatrix} \\
&\qquad\qquad +V*\left\langle\begin{pmatrix} u \\ v \\ w\\ \end{pmatrix},\,\pmb{\sigma}\begin{pmatrix} u \\ v \\ w\\ \end{pmatrix}\right\rangle_{\!\!\mathbb{C}^3}\!\!\bullet\,\pmb{\sigma}\begin{pmatrix} u \\ v \\ w \\\end{pmatrix},\end{split}
\end{equation}
from which it is not difficult to infer that the system is well-posed in $H^1(\mathbb{R}^3)\oplus H^1(\mathbb{R}^3)\oplus H^1(\mathbb{R}^3)$ if, for instance, $V^2\lesssim(\mathbbm{1}-\Delta)$ and $W^2\lesssim(\mathbbm{1}-\Delta)$ in the sense of forms of operators on $L^2(\mathbb{R}^3)$. The $H^1$-solution conserves the energy, namely the Hartree functional
\begin{equation}\label{eq:Hfunctional}
\begin{split}\mathcal{E}^{\mathrm{H}}\!\left[ \!\begin{pmatrix} u \\ v \\ w\\\end{pmatrix}\!\right]\;&:=\;\|\nabla u\|_{L^2(\mathbb{R}^3)}^2+\|\nabla v\|_{L^2(\mathbb{R}^3)}^2+\|\nabla w\|_{L^2(\mathbb{R}^3)}^2\\
&\quad +\frac{1}{2}\left \langle\begin{pmatrix} u \\ v \\ w\end{pmatrix},W*\big(|u|^2+|v|^2+|w|^2 \big)\begin{pmatrix} u \\ v \\ w \end{pmatrix}\right\rangle_{\!\mathfrak{h}}\\
&\quad +\frac{1}{2}\left \langle\begin{pmatrix} u \\ v \\ w\end{pmatrix}, V*\left\langle\begin{pmatrix} u \\ v \\ w \end{pmatrix},\,\pmb{\sigma}\begin{pmatrix} u \\ v \\ w \end{pmatrix}\right\rangle_{\!\!\mathbb{C}^3}\!\!\bullet\,\pmb{\sigma}\begin{pmatrix} u \\ v \\ w \end{pmatrix}\right\rangle_{\!\mathfrak{h}}.\end{split}
\end{equation}

Our main result takes the following form.

\begin{theorem}\label{thm:MFthm} Assume the following.
	\begin{itemize}
		\item[(i)] The potentials $V:\mathbb{R}^3\to\mathbb{R}$ and $W:\mathbb{R}^3\to\mathbb{R}$ are spherically symmetric and satisfy $V^2\lesssim(\mathbbm{1}-\Delta)$ and $W^2\lesssim(\mathbbm{1}-\Delta)$ in the sense of forms of operators on $L^2(\mathbb{R}^3)$.
		\item[(ii)] The normalised one-body spinor
		\begin{equation*}
		\phi_0\,=\,\begin{pmatrix} u_0 \\ v_0 \\ w_0 \\\end{pmatrix}\,,\qquad \|\phi_0\|_{\mathfrak{h}}\;=\;1
		\end{equation*}
		is given for some $u_0,v_0,w_0\in H^1(\mathbb{R}^3)$.
		\item[(iii)] For each $N\in\mathbb{N}$, $N\geqslant 2$, the initial many-body vector state $\Psi_N\in\cH_{N,\mathrm{sym}}$ satisfies $\|\Psi_N\|_{\cH_N}=1$ and exhibits complete spinor BEC onto $\phi_0$ in the quantitative sense
		\begin{equation}\label{eq:alpha_at_tzero}
		\alpha_{(\Psi_{\!N},\phi_0)}\;\leqslant\;\frac{K}{N}
		\end{equation}
		for the indicator defined in \eqref{eq:defalpha}, for some $K>0$.
		%  
		%  \[
		%   \lim_{N\to\infty}\gamma_{N}^{(1)}\;=\;\left|\!\begin{pmatrix} u_0 \\ v_0 \\ w_0\end{pmatrix}\!\right\rangle\left\langle\!\begin{pmatrix} u_0 \\ v_0 \\ w_0\end{pmatrix}\!\right|\,.
		%  \]
	\end{itemize}
	Correspondingly,
	\begin{itemize}
		\item let $t\mapsto\Psi_{\!N,t}$ be the solution to the Cauchy problem \eqref{eq:Cauchy_problem} with the initial datum $\Psi_{\!N}$ and with the many-body Hamiltonian $H_N$ defined in \eqref{eq:meanfieldHN} through the potentials $V$ and $W$ fixed by assumption (i), that is, $\Psi_{N,t}=e^{-\ii t H_N}\Psi_N$;
		\item let $t\mapsto\phi_t=\begin{pmatrix} u_t \\ v_t \\ w_t \end{pmatrix}\in\mathfrak{h}$ be the solution, with values in $H^1(\mathbb{R}^3)\otimes\mathbb{C}^3$, to the Cauchy problem consisting of the system \eqref{eq:MF_system} with the potentials $W$ and $V$ given by assumption (i) and with the initial datum $\phi_0$.
		%\item for every $t>0$ let $\alpha_{\,\Psi_{\!N,t}}:=1-\langle \phi_t,\gamma_{\Psi_{\!N,t}}\phi_t\rangle_\mathfrak{h}$.
	\end{itemize}
	Then, for every $t>0$, one has
	\begin{equation}\label{eq:alpha_at_t}
	\alpha_{(\Psi_{\!N,t},\phi_t)}\;\leqslant\;\frac{K+1}{N}\,e^{\,Ct}
	\end{equation}
	and hence also
	\begin{equation}\label{eq:trace_at_t}
	\mathrm{Tr}\,\big|\gamma_{N,t}^{(1)}-\,|\phi_t\rangle\langle\phi_t|\,\big|\;\lesssim\;\frac{1}{\sqrt{N\,}}\,e^{\,C t/2}
	\end{equation}
	for some constant $C>0$ that depends only on $W$, $V$, and $\phi_0$.
\end{theorem}

In view of formulas \eqref{eq:initialgamma}-\eqref{eq:equivalent-BEC-control} and of the discussion made in Section \ref{sec:intro}, Theorem \ref{thm:MFthm} indeed proves the persistence in time of spinor BEC for a spin-1 Bose gas described by the mean-field model \eqref{eq:meanfieldHN}, provided spinor BEC is present at the initial time, and it does that with a quantitative control in $N$, which are going to turn into physical units in Section \ref{sect:phys}.

For the type of technique used in the proof, that we shall present in Section \ref{sect:proof}, Theorem \ref{thm:MFthm} can be established also in various modified versions where the one-body kinetic energy operator $-\Delta$ is replaced by a more general one-body Schr\"{o}dinger operator $h=-\Delta+U_{\mathrm{trap}}$, where $U_{\mathrm{trap}}:\mathbb{R}^3\to\mathbb{R}$ is a suitable confining potential, or by the semi-relativistic kinetic energy operator $h=\sqrt{\mathbbm{1}-\Delta}$, and the like. When such a generic $h$ is included in the model, the main required modifications concern on the one hand the assumption on the potentials $V$ and $W$, that must guarantee a control of the form
\[
\|V^2*|f|^2\|_{L^\infty(\mathbb{R}^3)}\;\lesssim\;\|f\|^2_{\mathcal{D}[h]}
\]
for every $f$ in the form domain $\mathcal{D}[h]$ of $h$, in analogy with assumption (i) of Theorem \ref{thm:MFthm} where $\mathcal{D}[h]=H^1(\mathbb{R}^3)$, and on the other hand the well-posedness of the corresponding Hartree system in the `energy' space $\mathcal{D}[h]$. Apart from that, the specific type of $h$ does not play a substantial role in the proof, since $h$ drops out `intrinsically' when one computes the time derivative of $\alpha_{(\Psi_{\!N,t},\phi_t)}$ and for this reason we will only treat explicitly the concrete choice of the model \eqref{eq:meanfieldHN}.

As mentioned in Section \ref{sec:intro}, the purely mean-field scaling in the Hamiltonian \eqref{eq:meanfieldHN} can be considered as a first, somewhat rough model for the spinor Bose gas -- nevertheless we shall comment in Section \ref{sect:phys} that the consistence with the experimental data is already quite satisfactory.

One can make the scaling more realistic by replacing the two-body interactions $\frac{1}{N}W(x_1-x_2)$ and $\frac{1}{N}V(x_1-x_2)$ with short-scale interactions of the form
\begin{equation}
\begin{split}
W_N(x_1-x_2)\;&:=\;N^2W(N(x_1-x_2))\,, \\
V_N(x_1-x_2)\;&:=\;N^2V(N(x_1-x_2))\,,
\end{split}
\end{equation}
the so-called Gross-Pitaevskii scaling \cite{LSeSY-ober,am_GPlim}. Now $W$ and $V$ are chosen so that when $N=N_\mathrm{exp}$, the actual number of particles in the experiment, the scattering lengths of $W_{N_\mathrm{exp}}$ and of $V_{N_\mathrm{exp}}$ equal respectively the experimental values for the quantities $c_0$ and $c_2$ introduced \eqref{eq:scatt_lenghts_comb}. 
It is simple to deduce that the scattering lengths of $W_N$ and $V_N$ amount, respectively, to $c_{W_N}=(N_\mathrm{exp}/N)c_0$ and $c_{V_N}=(N_\mathrm{exp}/N)c_2$, which shows that the limit $N\to +\infty$ describes a regime of high dilution (the scattering lengths are dominated by the mean inter-particle distance $\sim \!N^{-1/3}$). Moreover, this scaling preserves the product (density of the gas)$\times$(scattering length), indeed $N c_{W_N}=\mathrm{const.}$ and $N c_{V_N}=\mathrm{const.}$, and hence the total energy per particle remains constant for many-body states that are close to the ground state of the low-density Bose gas in a trap \cite{LY-1998}.

The resulting $N$-body spinor Hamiltonian with Gross-Pitaevskii scaling is
\begin{equation}\label{eq:GP-HN}
\begin{split}
H_N^{\mathrm{GP}}\;&:=\;\sum_{j=1}^N(-\Delta_{x_j})+N^2\!\!\!\sum_{1<j\leqslant k<N}W(N(x_j-x_k)) \\
&\qquad\qquad+N^2\!\!\!\sum_{1<j\leqslant k<N}V(N(x_j-x_k))\,\pmb{\sigma}_j\bullet\pmb{\sigma}_k\,,
\end{split}
\end{equation}
and by formally replacing $W(x)\rightarrow 8\pi c_0\delta(x)$ and $V(x)\rightarrow 8\pi c_2\delta(x)$ one can argue that at the effective level the spinor Hartree system is replaced by the spinor Gross-Pitaevskii system
\begin{equation}\label{eq:GP_system}
\begin{split}
\ii\partial_t u\;&=\;-\Delta u+ 8\pi c_0(|u|^2+|v|^2+|w|^2)u \\
&\qquad\quad+8\pi c_2(|u|^2+|v|^2-|w|^2)u+8\pi c_2(\overline{w}v)v \\
\ii\partial_t v\;&=\;-\Delta u+ 8\pi c_0(|u|^2+|v|^2+|w|^2)v \\
&\qquad\quad+8\pi c_2(|u|^2+|w|^2)v+16\pi c_2(\overline{v}w)u  \\
\ii\partial_t w\;&=\;-\Delta u+ 8\pi c_0(|u|^2+|v|^2+|w|^2)w \\
&\qquad\quad+8\pi c_2(|v|^2+|w|^2-|u|^2)w+8\pi c_2(\overline{u}\,\overline{v})v\,,
\end{split}
\end{equation}
where $c_0 N_\mathrm{exp}$ and $c_2 N_\mathrm{exp}$ are the scattering lengths of the unscaled potentials $W$ and $V$ respectively.

The system \eqref{eq:GP_system} provides indeed the emergent dynamics observed in the experiments -- see, e.g., \cite[Eq.~(1)-(3)]{Chang-Qin-Zhang-You-2005}.

The mathematical treatment of the Gross-Pitaevskii scaling is more laborious because it requires an efficient control of the short-scale structure that is induced in the course of the evolution by the Hamiltonian $ H_N^{\mathrm{GP}}$ at the spatial scale $N^{-1}$ \cite{EMS-2008}. This structure carries a negligible contribution to the $\cH$-norm of the many-body state locally in space, but a significant contribution to its total energy. For this reason one must restrict the analysis to many-body states that not only exhibit spinor BEC, but also have the correct scale of energy per particle, measured with respect to the one-body Gross-Pitaevskii spinor energy functional defined as
\begin{equation}\label{eq:GPfunctional}
\begin{split}
\mathcal{E}^{\mathrm{GP}}\!\left[ \!\begin{pmatrix} u \\ v \\ w\end{pmatrix}\!\right]\;&:=\;\|\nabla u\|_{L^2(\mathbb{R}^3)}^2+\|\nabla v\|_{L^2(\mathbb{R}^3)}^2+\|\nabla w\|_{L^2(\mathbb{R}^3)}^2\\
&\quad +4\pi c_0 \left \langle\begin{pmatrix} u \\ v \\ w\end{pmatrix},\big(|u|^2+|v|^2+|w|^2 \big)\begin{pmatrix} u \\ v \\ w \end{pmatrix}\right\rangle_{\!\mathfrak{h}}\\
&\quad +4\pi c_2\left \langle\begin{pmatrix} u \\ v \\ w\end{pmatrix}, \left\langle\begin{pmatrix} u \\ v \\ w \end{pmatrix},\,\pmb{\sigma}\begin{pmatrix} u \\ v \\ w \end{pmatrix}\right\rangle_{\!\!\mathbb{C}^3}\!\!\bullet\,\pmb{\sigma}\begin{pmatrix} u \\ v \\ w \end{pmatrix}\right\rangle_{\!\mathfrak{h}}.
\end{split}
\end{equation}

From the technical point of view, the meticulous analysis we made to prove the emergence of the Gross-Pitaevskii dynamics for pseudo-spinor condensates \cite{MO-pseudospinors-2017} can be extended to those steps of the present proof of Theorem \ref{thm:MFthm} that are new because of the spinor interaction, so as to produce the following analogous version of Theorem \ref{thm:MFthm}. We content ourselves to state it here without proof.

\begin{theorem}\label{thm:GPthm} Assume the following.
	\begin{itemize}
		\item[(i)] The potentials $V:\mathbb{R}^3\to\mathbb{R}$ and $W:\mathbb{R}^3\to\mathbb{R}$ belong to $L^\infty(\mathbb{R}^3)$, are compactly supported, spherically symmetric, and non-negative. Let $g_0$ and $g_2$ be their respective scattering lengths.
		\item[(ii)] The normalised one-body spinor
		\[
		\phi_0\,=\,\begin{pmatrix} u_0 \\ v_0 \\ w_0 \end{pmatrix}\,,\qquad \|\phi_0\|_{\mathfrak{h}}\;=\;1
		\]
		is given for some $u_0,v_0,w_0\in H^2(\mathbb{R}^3)$ so that the Cauchy problem consisting of the system \eqref{eq:GP_system} with the couplings $g_0$ and $g_2$ given in assumption (i) and with initial datum $\phi_0$ admits a unique local-in-time solution $t\mapsto\phi_t=\begin{pmatrix} u_t \\ v_t \\ w_t \end{pmatrix}$ with values in $H^2(\mathbb{R}^3)\otimes\mathbb{C}^3$.
		\item[(iii)] For each $N\in\mathbb{N}$, $N\geqslant 2$, the initial many-body vector state $\Psi_{\!N}\in\cH_{N,\mathrm{sym}}$ is given in the domain of $H_N^{\mathrm{GP}}$, satisfying $\|\Psi_{\!N}\|_{\cH_N}=1$, such that $\Psi_N$ exhibits complete spinor BEC onto $\phi_0$ in the quantitative sense
		\[
		\mathrm{Tr}\,\big|\gamma_{\Psi_{\!N}}^{(1)}-\phi_0\rangle\langle\phi_0|\,\big|\;\lesssim\;\frac{1}{\,N^{\eta_1}}
		%\left|\!\begin{pmatrix} u_0 \\ v_0 \\ w_0\end{pmatrix}\!\right\rangle\left\langle\!\begin{pmatrix} u_0 \\ v_0 \\ w_0\end{pmatrix}\!\right|\,.
		\]
		for some $\eta_1>0$, and has the following asymptotics of the energy per particle
		\[
		\Big| \frac{1}{N}\langle\Psi_{\!N},H_N^{\mathrm{GP}}\Psi_{\!N}\rangle_{\cH_N}-\mathcal{E}^{\mathrm{GP}}[\phi_0]\Big|\;\lesssim\;\frac{1}{\,N^{\eta_2}}
		\]
		for some $\eta_2>0$. 
	\end{itemize}
	Correspondingly, let $t\mapsto\Psi_{\!N,t}$ be the solution to the Cauchy problem \eqref{eq:Cauchy_problem} with the initial datum $\phi_0$ and with the many-body Hamiltonian $H_N^{\mathrm{GP}}$ defined in \eqref{eq:GP-HN} through the potentials $V$ and $W$ fixed by assumption (i), that is, $\Psi_{N,t}=e^{-\ii t H_N}\Psi_N$.
	Then, for every $t>0$ for which $\phi_t$ exists, one has
	\[
	\lim_{N\to\infty}\gamma_{\Psi_{\!N,t}}^{(1)}\;=\;|\phi_t\rangle\langle\phi_t|.
	%\left|\!\begin{pmatrix} u_t \\ v_t \\ w_t\end{pmatrix}\!\right\rangle\left\langle\!\begin{pmatrix} u_t \\ v_t \\ w_t\end{pmatrix}\!\right|\,.
	\]
\end{theorem}

As before, Theorem \ref{thm:GPthm} proves the persistence of spinor BEC for a spin-1 Bose gas, once this condition is present at initial time and in the appropriate (ground-state-like) energy manifold. It also provides a quantitative rate of convergence in $N$ for $\gamma_{\Psi_{\!N,t}}^{(1)}$ that it is not particularly informative to make explicit, since it would be a surely non-optimal inverse power $N^{-\eta}$ for some small $\eta>0$ that depends on $\phi_0$, $W$, $V$, $\eta_1$, and $\eta_2$.

As commented already for Theorem \ref{thm:MFthm}, the technique we adopt here allows one to generalise Theorem \ref{thm:GPthm} with more general one-body Schr\"{o}dinger operators replacing the one-body kinetic operator $-\Delta$. For example, the details for dealing with the \emph{magnetic} Laplacian $(\ii\nabla+\mathbf{A})^2$ in the analogous theorem for scalar bosons are worked out in \cite{AO-GP_magnetic_lapl-2016volume}.

Theorem \ref{thm:MFthm} in the pure mean-field case, and Theorem \ref{thm:GPthm} in the ameliorated mean-field-like scaling of Gross-Pitaevskii type have many precursors in the literature for scalar condensates, where the order parameter is a scalar wave function instead of a spinor.

For the vast, and by now classical literature on the uniqueness of the minimiser of the Gross-Pitaevskii functional, the leading order of the ground state energy, and the emergence of condensation in the ground state, in the case of scalar bosons, we refer to the monograph \cite{LSeSY-ober} and the references therein. More recently, based on quantum de Finetti methods \cite{CKMR-QdF-2007,Lewin-Nam-Rougerie-remQdF-2015} the ground state energy asymptotics and the proof of condensation in the ground state of a scalar Bose gas was re-obtained in \cite{Lewin-Nam-Rougerie-2014-HartreeTheory} in the mean-field scaling and in \cite{Nam-Rougerie-Seiringer_GPrevisited-2016} in the Gross-Pitaevskii scaling. The dynamical analysis too has covered different space dimensions ($d=1,2,3$), a wide range of local singularities and long-distance decays for the inter-particle interactions, and various types of scaling limits in the many-body Hamiltonian $H_N$. We refer to the reviews \cite{S-2007,S-2008,Benedikter-Porta-Schlein-2015} for a comprehensive outlook, remarking that over the years the dynamical problem has involved a variety of approaches and techniques from analysis, operator theory, kinetic theory, and probability.

In the present work we employ a particularly robust and versatile method, invented and refined for \emph{scalar} Bose gases by P.~Pickl \cite{Pickl-JSP-2010,Pickl-LMP-2011,Pickl-RMP-2015}, with the contribution of A.~Knowles \cite{kp-2009-cmp2010}, which monitors the smallness of the indicator $\alpha_{(\Psi_{\!N,t},\phi_t)}$ by means of an ad hoc ``counting'' of the amount of particles in the many-body state $\Psi_{\!N,t}$ that `occupy' the one-body state $\phi_t$.

We extended this approach for the first time to pseudo-spinor condensates in our recent work \cite{MO-pseudospinors-2017}. The spinor Hamiltonian \eqref{eq:meanfieldHN} gives rise to terms that are qualitatively similar to those of the pseudo-spinor model with a Rabi coupling to an external magnetic field, plus \emph{additional new terms} that are due to the spin-spin interaction. Our analysis will primarily focus on such terms.

\section{Proof of Theorem \ref{thm:MFthm}} \label{sect:proof}

In this Section we present the proof of Theorem \ref{thm:MFthm}.

The main bound \eqref{eq:alpha_at_t} is going to be established by means of a Gr\"onwall argument, thus controlling the time derivative of $\alpha_{(\Psi_{\!N,t},\phi_t)}$ in terms of $\alpha_{(\Psi_{\!N,t},\phi_t)}$ itself. To do so, it is convenient to re-write $\alpha_{(\Psi_{\!N,t},\phi_t)}$ in a way that allows good algebraic manipulations (see \eqref{eq:identities}). This is the key idea of the so-called `counting' technique, first developed in \cite{kp-2009-cmp2010,Pickl-LMP-2011} and later refined in \cite{Pickl-RMP-2015} to cover the Gross-Pitaevskii case.

Before the actual proof, let us introduce an amount of definitions, notation, and auxiliary results.

\subsection{Preparatory material}

Throughout this Section, by $f\lesssim g$ we mean that $f\leqslant C g$ for some constant $C=C(W,V,\phi_0)>0$ independent of $t$ or $N$. To denote more restrictive dependences we shall write $f \lesssim_{W,V} g$ and the like.

One-body or two-body operators on $\cH_N$ will be denoted by $A_j$ and $B_{ij}$ to indicate that they act non-trivially, respectively as the operator $A$ on $\mathfrak{h}$ and the operator $B$ on $\mathfrak{h}\otimes \mathfrak{h}$,  only on the degrees of freedom of the $j$-th, or of the $i$-th and $j$-th particles.

In particular, the orthogonal projections
\begin{equation}
p_t\;:=\;\left|\!\begin{pmatrix} u_t \\ v_t \\ w_t\end{pmatrix}\!\right\rangle\left\langle\!\begin{pmatrix} u_t \\ v_t \\ w_t\end{pmatrix}\right|\,,\qquad q_t\;:=\;\mathbbm{1}-p_t
% 
% 
% \Bigg|\begin{pmatrix} u_t \\ v_t \\ w_t \end{pmatrix}\Bigg\rangle\Bigg\langle \begin{pmatrix} u_t \\ v_t \\ w_t \end{pmatrix}\Bigg|,\qquad q_t:=\mathbbm{1}-p_t,
\end{equation}
onto the spinor solution $\phi_t=\begin{pmatrix} u_t \\ v_t \\ w_t \end{pmatrix}$ to the system \eqref{eq:MF_system} and onto its orthogonal complement lift to the operators
\begin{equation}\label{eq:pi_qi}
\begin{split}
(p_t)_j\;=\;&\underbrace{\mathbbm{1}\otimes\dots \otimes\mathbbm{1}}_{j-1}\otimes\, p_t\otimes\underbrace{\mathbbm{1}\otimes\dots\otimes\mathbbm{1}}_{N-j}\\
(q_t)_j\;=\;&\underbrace{\mathbbm{1}\otimes\dots \otimes\mathbbm{1}}_{j-1}\otimes\, q_t\otimes\underbrace{\mathbbm{1}\otimes\dots\otimes\mathbbm{1}}_{N-j}\,,\qquad j\in\{1,\dots,N\}
\end{split}
\end{equation}
on  $\mathcal{H}_{N}$.

We define the orthogonal projections $P_k$, $k\in\mathbb{Z}$, by
\begin{equation} \label{eq:def_P}
\begin{array}{ll}
P_k\;:=\!\displaystyle\sum_{a\in\{0,1\}^N\atop\sum_{i=1}^Na_i=k}\bigotimes_{i=1}^N\;(p_t)_i^{1-a_i}(q_t)^{a_i}_i\quad & \textrm{if }k\in\{0,\dots,N\} \\
P_k\;:=\;\mathbbm{O}\,, & \textrm{otherwise}\,.
\end{array}
\end{equation}
It follows from \eqref{eq:def_P} that
\begin{equation} \label{eq:completeness}
\big[P_j,P_k \big]\;=\;\delta_{jk}P_k,\qquad \sum_{k=0}^NP_k\;=\;\mathbbm{1}.
\end{equation}
The Hilbert subspace of $\cH_N$ which $P_k$ projects onto is naturally interpreted as the space of $N$-body states with exactly $k$ particles `out of the condensate', in the sense of orthogonality with respect to the spinor $\phi_t$.

Next, we define
%the bounded operator on $\mathcal{H}_{N,\mathrm{sym}}$
\begin{equation}\label{eq:def_mhat_nhat}
\begin{array}{ll}
\displaystyle \widehat m\;:=\;\sum_{k=0}^Nm(k)P_k \qquad & \displaystyle m(k)\;:=\;\frac{k}{N} \\
\displaystyle \widehat n\;:=\;\sum_{k=0}^Nm(k)P_k \qquad & \displaystyle n(k)\;:=\;\sqrt{\frac{k}{N}}
\end{array}
\end{equation}
and their `shifted' counterparts
\begin{equation}
\widehat{m_d}\;:=\;\sum_{k=0}^{N}m(k+d)P_k\,,\qquad \widehat{n_d}\;:=\;\sum_{k=0}^{N}n(k+d)P_k\,,\qquad d\in\mathbb{Z}\,.
\end{equation}
$P_k$, $\widehat{m_d}$, $\widehat{n_d}$ clearly depend on time through $\phi_t$, although for a lighter notation we omit such a dependence.

A key observation is that the indicator $\alpha_{(\Psi_{\!N,t},\phi_t)}$ introduced in \eqref{eq:defalpha} can be re-expressed in terms of the operators \eqref{eq:pi_qi}, \eqref{eq:def_P}, and \eqref{eq:def_mhat_nhat} as
\begin{equation} \label{eq:identities}
\alpha_{(\Psi_{\!N,t},\phi_t)}\;=\; \langle \Psi_{\!N,t},(q_t)_1\Psi_{\!N,t}\rangle_{\mathcal{H}_N}\;=\; \langle \Psi_{\!N,t},\widehat m \Psi_{\!N,t}\rangle_{\mathcal{H}_N}\,.
\end{equation}
Indeed, owing to the bosonic symmetry of $\Psi_{\!N,t}$ and the property \eqref{eq:completeness},
\begin{equation*} 
\begin{split}
\langle \Psi_{\!N,t},(q_t)_1&\Psi_{\!N,t}\rangle_{\mathcal{H}_N}\;=\;\frac{1}{N}\sum_{j=1}^N\langle\Psi_{\!N,t},(q_t)_j\Psi_{\!N,t}\rangle_{\mathcal{H}_N}\\
&\; =\;\frac{1}{N}\sum_{j=1}^N\sum_{k=0}^N\langle\Psi_{\!N,t},(q_t)_jP_k\Psi_{\!N,t}\rangle_{\mathcal{H}_N}\;=\;\frac{1}{N}\sum_{k=0}^Nk\,\langle\Psi_{\!N,t},P_k\Psi_{\!N,t}\rangle\,.
\end{split}
\end{equation*}
In fact, \eqref{eq:identities} will allow for very convenient algebraic manipulations.

%The expectation of $\widehat{m}$ on $\Psi_{\!N,t}$ is interpreted as a collective quantification of the components of the state $\Psi_{\!N,t}$ in which $k$ particles are not in the condensate spinor $\phi_t$, for $k=0,\dots,N$; each such component is weighted $k/N$ times the square of its norm. 

Another useful property is the following, which we shall use systematically in the proof in order to compute commutators involving $\widehat{m}$, whose proof may be found, e.g., in \cite[Lemma 3.10]{kp-2009-cmp2010}.

\begin{lemma}[Commutation property] \label{lemma:commutation}
	Let $Q_a$ and $Q_b$ be two tensor-product monomials in $p_t$ and $q_t$, and let $B_{12}$ be a two-body operator on $\cH_N$, all three operators acting non-trivially only on the degrees of freedom of the first and second particle. Then
	\begin{equation*}
	\begin{split}
	Q_a \,C_{12}\,\widehat{m}\,Q_b\;&=\;Q_a \,\widehat{m_d}\,C_{12}\,Q_b \\
	\big[\widehat{m_d},(p_t)_j\big]\;&=\;\big[\widehat{m_d},(q_t)_j\big]\;=\;\mathbbm{O} 
	\end{split}
	\end{equation*}
	and 
	\begin{equation*}
	\begin{split}
	Q_a \,C_{12}\,\widehat{n}\,Q_b\;&=\;Q_a \,\widehat{n_d}\,C_{12}\,Q_b \\
	\big[\widehat{n_d},(p_t)_j\big]\;&=\;\big[\widehat{n_d},(q_t)_j\big]\;=\;\mathbbm{O}
	\end{split}
	\end{equation*}
	for any $j\in\{1,\dots,N\}$ and $d\in\mathbb{Z}$, where $d\in\mathbb{Z}$ is the difference between the number of $q_t$'s in $Q_b$ and the number of $q_t$'s in $Q_a$.
\end{lemma}

Further relevant properties of $\widehat m$ and $\widehat n$ are collected in the next Lemma (see, e.g., \cite[Section 3]{kp-2009-cmp2010} for the standard proof).

\begin{lemma}[Properties of $\widehat m$ and $\widehat n$] \label{lemma:mn}
	\begin{itemize}\item[]
		\item[(i)] For any $d\in\mathbb{Z}$,
		\begin{equation*}
		\widehat{m}-\widehat{m_{d}}\;=\;-\frac{d}{N}\mathbbm{1}\,.
		\end{equation*}
		\item[(ii)] The operators $\widehat{m}$ and $\widehat{n}$ are invertible on the range of $(q_t)_1$ with inverse that is bounded on $\cH_N$. For such inverses we shall write $\widehat{m}^{-1}(q_t)_1$ and $\widehat{n}^{-1}(q_t)_1$.
		\item[(iii)] For any $\Psi\in \mathcal{H}_{N,\mathrm{sym}}$,
		\begin{equation*}
		\langle\Psi, \widehat m^{-1}(q_t)_1(q_t)_2\Psi\rangle_{\mathcal{H}_N}\;\leqslant \;\frac{N}{N-1}\,\langle\Psi,\widehat m \Psi\rangle_{\mathcal{H}_N}\,.
		\end{equation*}
	\end{itemize}
\end{lemma}

% 
% \begin{proof}
% 	(i) follows directly from \eqref{eq:completeness}.
% 	
% 	(ii) Upon defining
% 	\begin{equation*}
% 	\begin{split}
% 	\widehat{m}^{-1}(q_t)_1\Psi\;&:=\;\sum_{k=1}^N\frac{N}{k}\,P_k\,(q_t)_1\Psi \\
% 	\widehat{n}^{-1}(q_t)_1\Psi\;&:=\;\sum_{k=1}^N\sqrt{\frac{N}{k}}\,P_k\,(q_t)_1\Psi\,,
% 	}
% 	\end{equation*}
% 	one immediately checks that the above expressions define vectors in $\mathcal{H}_N$ and the identities $\widehat{m}\widehat{m}^{-1}(q_t)_1\Psi=(q_t)_1\Psi$ and $\widehat{n}\widehat{n}^{-1}(q_t)_1\Psi=(q_t)_1\Psi$ follow by direct computation using \eqref{eq:completeness}.
% 	
% 	(iii) follows from
% 	\begin{equation*}
% 	\begin{split}
% 	\langle\Psi, \widehat m^{-1}(q_t)_1(q_t)_2\Psi\rangle&\;=\;\frac{1}{N-1}\sum_{k=2}^{N}\langle\Psi,\widehat{m}^{-1}(q_t)_1 (q_t)_k\Psi\rangle\\
% 	&\;\leqslant\;\frac{N}{N-1}\,\frac{1}{N}\sum_{k=1}^{N}\langle\Psi,\widehat{m}^{-1}(q_t)_1 (q_t)_k\Psi\rangle\\
% 	&\;\leqslant\;\frac{N}{N-1}\,\langle\Psi,(q_t)_1\Psi\rangle\;=\;\frac{N}{N-1}\,\langle\Psi,\widehat{m}\Psi\rangle\,,
% 	}
% 	\end{equation*}
% 	where the first identity is due to the fact that $(q_t)_1\Psi$ is symmetric under permutation of the degrees of freedom of the last $N-1$ particles, the second step is due to the positivity of $\widehat{m}^{-1}(q_t)_1$, and the last two steps are based, respectively, on \eqref{eq:def_mhat_nhat} and \eqref{eq:identities}.
% \end{proof}

The catch from Lemma \ref{lemma:mn} is: $\widehat{m}$ and its `shifted' version differ only by a sub-leading term, and moreover, in the expectation on many-body bosonic states, one can effectively replace two $q$'s (acting on different particles) with two $\widehat{m}$'s. For short we shall also write $\widehat{m}^{-1}$ and $\widehat{n}^{-1}$ instead of $\widehat{m}^{-1}(q_t)_1$ and $\widehat{n}^{-1}(q_t)_1$, understanding that such operators act on the range of $(q_t)_1$.

Last, we need a control on the norm of five relevant operators on $\mathfrak{h}$, which we denote by $W^{\phi_t}$, $V^{\phi_t}$, $D_{W,\phi_t}$, $E_{V,\phi_t}$, and $F_{W,V,\phi_t}$, and may be thought of as operators of (scalar or spinor) multiplication by the following `mean-field' (or `smeared') potentials:
\begin{equation} \label{eq:smeared potentials}
\begin{split}
W^{\phi_t}(x)\;&:=\;W*(|u_t|^2+|v_t|^2+|w_t|^2)\\
V^{\phi_t}(x)\;&:=\;(V*\langle\phi_t,\pmb{\sigma}\phi_t\rangle_{\mathbb{C}^3})\bullet\pmb\sigma\\
D_{W,\phi_t}(x)\;&:=\;W^2*(|u_t|^2+|v_t|^2+|w_t|^2) \\
%\langle\phi_t,W^2(x-\cdot)\phi_t\rangle_\mathfrak{h}\\
E_{V,\phi_t}(x)\;&:=\;\langle\phi_t(\cdot),V^2(x-\cdot)(\pmb\sigma_\cdot\bullet\pmb\sigma)^2\phi_t(\cdot)\rangle_\mathfrak{h}\\
F_{W,V,\phi_t}(x)\;&:=\;\langle\phi_t(\cdot),\big(W(x-\cdot)+V(x-\cdot)\pmb\sigma_\cdot\bullet\pmb\sigma\big)^2\phi_t(\cdot)\rangle_\mathfrak{h}\,.
\end{split}
\end{equation}
Thus, $W^{\phi_t}$ and $D_{W,\phi_t}$ multiply, respectively, by the scalar functions $W^{\phi_t}(x)$ and $D_{W,\phi_t}(x)$, which can be re-written, in terms of the rule for the scalar product in $\mathfrak{h}$, also as $W^{\phi_t}(x)=\langle\phi_t(\cdot),W(x-\cdot)\phi_t(\cdot)\rangle_\mathfrak{h}$ and $D_{W,\phi_t}(x)=\langle\phi_t(\cdot),W^2(x-\cdot)\phi_t(\cdot)\rangle_\mathfrak{h}$. In the treatment of \emph{pseudo-spinor} condensates only these two types of operators are present: the other three are specific for the present treatment of spinor condensates. $V^{\phi_t}$ acts by a matrix multiplication, according to the rule for the $\bullet$-product, namely $(V^{\phi_t}\psi)(x)=\sum_{j=1}^3(V*\langle\phi_t,\sigma^{(j)}\phi_t\rangle_{\mathbb{C}^3})(x)(\sigma^{(j)}\psi)(x)$. Analogously, $E_{V,\phi_t}$ and $F_{W,V,\phi_t}$ too act non-trivially on both the space and the spin degrees of freedom of $\mathfrak{h}$: in particular, they act on the spin degrees of freedom through the $\pmb\sigma$ that stays in the second factor of the $\bullet$-product.

We shall use systematically the property, that we now prove, that the above operators are bounded on $\mathfrak{h}$ \emph{uniformly in time}.

%The following Lemma, systematically used during the proof, shows that the operators defined in \eqref{eq:smeared potentials} are bounded uniformly in time. %We remark that, being the three operators containing $V$ not present in the scalar case (and being them non-diagonal in the spin space), they need be dealt with differently from the others.

\begin{lemma}[Uniform-in-time boundedness of the smeared potentials] \label{lemma:smeared}
	One has
	\begin{equation}\label{eq:unifbdd}
	\big\|W^{{\phi_t}}\big\|_{\mathrm{op}},\;	\big\|V^{\phi_t}\big\|_{\mathrm{op}},\;\big\|D_{W,\phi_t}\big\|_{\mathrm{op}},\;\big\|E_{V,\phi_t}\big\|_{\mathrm{op}},\;\big\|F_{W,V,\phi_t}\big\|_{\mathrm{op}}\;\lesssim\; 1.
	\end{equation}
\end{lemma}
\begin{proof}
	We shall actually show that the above operator norms are all controlled by the square of the $H^1$-norm of $\phi_t$.
	After that, one observes that %the bounds
	%	\[
	%	 \big\|W^{{\phi_t}}\big\|_{\mathrm{op}}\;\lesssim_{W}\;\|\phi_t\|_{H^1(\mathbb{R}^3)\otimes \mathbb{C}^3}\quad\textrm{ and }\quad \big\|V^{{\phi_t}}\big\|_{\mathrm{op}}\;\lesssim_{V}\;\|\phi_t\|_{H^1(\mathbb{R}^3)\otimes \mathbb{C}^3}
	%	\]
	the Hartree energy functional \eqref{eq:Hfunctional}, rewritten as
	\begin{equation*} \label{eq:Hartree_functional}
	\mathcal{E}^{\mathrm{H}}[\phi_t]\;=\;\big\|\nabla\phi_t\big\|_{\mathfrak{h}}^2+\frac{1}{2}\langle\phi_t,\big( W^{\phi_t}+V^{\phi_t}\big)\phi_t\rangle_\mathfrak{h},
	\end{equation*}
	is conserved along the Hartree system \eqref{eq:MF_system}, and this implies, for every $\varepsilon>0$,
	\begin{equation*}
	\big\|\nabla\phi_t\big\|_\mathfrak{h}^2\;\leqslant\;\mathcal{E}^{\mathrm{H}}[\phi_0]+\varepsilon^{-1}+\varepsilon\big\|F_{W,V,\phi_t}\big\|_{\mathrm{op}}.
	\end{equation*}
	This, by the bound
	\begin{equation*}
	\big\|F_{W,V,\phi_t}\big\|_{\mathrm{op}}\;\lesssim_{W,V}\;\|\phi_t\|_{H^1(\mathbb{R}^3)\otimes\mathbb{C}}^2,
	\end{equation*}
	implies
	\begin{equation*}
	\|\phi_t\|_{H^1(\mathbb{R}^3)\otimes \mathbb{C}^3}\;\lesssim_{W,V}\;\|\phi_0\|_{H^1(\mathbb{R}^3)\otimes \mathbb{C}^3}
	\end{equation*}
	and hence the boundedness \eqref{eq:unifbdd} uniformly in time.

	Owing to the assumption (i) in Theorem \ref{thm:MFthm}\,,
	\begin{equation*} 
	W\;\leqslant\;{\textstyle\frac{1}{2}}(\mathbbm{1}+W^2)\;\lesssim_W\;\mathbbm{1}-\Delta\,,
	\end{equation*}
	whence $\|W^b*|f|^2\|_{L^\infty(\mathbb{R}^3)}\lesssim_W\|f\|_{H^1(\mathbb{R}^3)}^2$ for $b\in\{1,2\}$ and $f\in H^1(\mathbb{R}^3)$, and hence
	\begin{equation*}
	\big\|W^{\phi_t}\big\|_{\mathrm{op}}\;\lesssim_W\; \|\phi_t\|^2_{H^1(\mathbb{R}^3)\otimes \mathbb{C}^3} \quad\textrm{ and }\quad \big\|D_{W,\phi_t}\big\|_{\mathrm{op}}\;\lesssim_W\; \|\phi_t\|^2_{H^1(\mathbb{R}^3)\otimes \mathbb{C}^3} \,.
	\end{equation*}
	
	Concerning $V^{\phi_t}$, one has
	\[%\label{eq:estimateV}
	\big\|V^{\phi_t}\big\|_{\mathrm{op}} \;\leqslant\; \sum_{j=1}^{3} \big\| V*\langle\phi_t,\sigma^{(j)}\phi_t\rangle_{\mathbb{C}^3}\big\|_{L^\infty(\mathbb{R}^3)}\,.
	\]
	The $L^\infty$-norms in the above expression are of the form $\|V*(fg)\|_{L^\infty(\mathbb{R}^3)}$, where $f,g$ can be $u_t,v_t,w_t$, and writing $|fg|\leqslant\frac{1}{2}(|f|^2+|g|^2)$ one boils down the estimate to $\|V*|f|^2\|_{L^\infty(\mathbb{R}^3)}$: due again to the assumption (i) in Theorem \ref{thm:MFthm}, $V\lesssim_V(\mathbbm{1}-\Delta)$ and $\|V*|f|^2\|_{L^\infty(\mathbb{R}^3)}\lesssim_V\|f\|_{H^1(\mathbb{R}^3)}^2$, whence
	\[
	\big\|V^{\phi_t}\big\|_{\mathrm{op}}\;\lesssim_V\; \|\phi_t\|^2_{H^1(\mathbb{R}^3)\otimes \mathbb{C}^3}\,.
	\]

	The scheme of the estimate of the norm of $E_{V,\phi_t}$ and $F_{W,V,\phi_t}$ is completely analogous, for the orbitals $u_t,v_t,w_t$ are merely linearly re-shuffled by the action of the spinor part of the operator, namely by $\pmb{\sigma}$.
\end{proof}

\subsection{Gr\"onwall estimate for $\alpha$}

One can easily see that the function $t\mapsto \alpha_{(\Psi_{\!N,t},\phi_t)}$ defined in \eqref{eq:defalpha} is actually differentiable. Indeed, $\alpha_{(\Psi_{\!N,t},\phi_t)}= \langle \Psi_{\!N,t},(q_t)_1\Psi_{\!N,t}\rangle_{\mathcal{H}_N}$ (see \eqref{eq:identities} above), therefore when the time derivative hits the $\Psi_{\!N,t}$'s this produces the commutator $[H_N,(p_t)_1]$, owing to the many-body equation $\ii\partial_t \Psi_{\!N,t}=H_N\Psi_{\!N,t}$, and this term is well-defined because $(p_t)_1\Psi_{\!N,t}\in(H^1(\mathbb{R}^3)\otimes\mathbb{C}^3)^{\otimes N}$, which is the form domain of $H_N$. When instead the time derivative hits $(q_t)_1$, this produces the commutator $[-\Delta+W^{\phi_t}+V^{\phi_t},p_t]$, owing to the one-body non-linear equation \eqref{eq:MF_system_compact} and using the definition \eqref{eq:smeared potentials}: then, as a consequence of  Lemma \ref{lemma:smeared}, the operator $-\Delta+W^{\phi_t}+V^{\phi_t}$ maps continuously $H^1(\mathbb{R}^3)\otimes\mathbb{C}^3$ (its form domain) into its dual, which, together with $(p_t)_1\Psi_{\!N,t}\in(H^1(\mathbb{R}^3)\otimes\mathbb{C}^3)^{\otimes N}$, makes the expectation $\langle \Psi_{\!N,t},[(-\Delta+W^{\phi_t}+V^{\phi_t})_1,(p_t)_1]\Psi_{\!N,t}\rangle_{\cH_N}$ well defined too.

We can then differentiate time and exploit the bosonic symmetry of $\Psi_{\!N,t}$, repeating the manipulations that led to \eqref{eq:identities}: we obtain
\begin{equation} \label{eq:time_derivative_alpha}
\begin{split}
\dot\alpha_{(\Psi_{\!N,t},\phi_t)}\;=\;\frac{\ii}{2}\big\langle \Psi_{\!N,t},\big[&  (N-1)W_{12} +(N-1)V_{12}\,\pmb{\sigma}_1\bullet\pmb{\sigma}_2 \\ 
&- NW^{\phi_t}_1-NW^{\phi_t}_2-N V_1^{\phi_t}-NV^{\phi_t}_2, \widehat{m}\big]\,\Psi_{\!N,t}\big\rangle_{\mathcal{H}_N}\,.
\end{split}
\end{equation}
It is worth stressing that an important cancellation occurred in \eqref{eq:time_derivative_alpha}, the r.h.s.~of which does not depend on the kinetic operator $-\Delta$ any longer.

Now, inserting
\begin{equation*}
\mathbbm{1}\;=\;\big((p_t)_1+(q_t)_1\big)\big((p_t)_2+(q_t)_2\big)
\end{equation*}
in both sides of the commutator in \eqref{eq:time_derivative_alpha} and expanding the products produces 16 terms: those with the \emph{same number} of $q_t$'s to the left and to the right vanish, due to Lemma \ref{lemma:commutation}, and the remaining ones are all of the form
\begin{equation} \label{eq:three_terms}
\begin{split}
(\mathrm{I})\;:=\;\frac{\ii}{2}\big\langle\Psi_{\!N,t},& (p_t)_1(p_t)_2 \big[(N-1)W_{12}+(N-1)V_{12}\,\pmb\sigma_1\bullet\pmb\sigma_2\\
&\qquad\quad\quad\quad-NW_1^{\phi_t}-NV_1^{\phi_t},\widehat{m}\big](q_t)_1(p_t)_2\Psi_{\!N,t}\big\rangle_{\mathcal{H}_N}\\
(\mathrm{II})\;:=\;\frac{\ii}{2}\big\langle\Psi_{\!N,t},& (p_t)_1(q_t)_2 \big[(N-1)W_{12}+(N-1)V_{12}\,\pmb\sigma_1\bullet\pmb\sigma_2\\
&\qquad\quad\quad\quad-NW_1^{\phi_t}-NV_1^{\phi_t},\widehat{m}\big](q_t)_1(q_t)_2\Psi_{\!N,t}\big\rangle_{\mathcal{H}_N}\\
(\mathrm{III})\;:=\;\frac{\ii}{2}\big\langle\Psi_{\!N,t},& (p_t)_1(p_t)_2 \big[(N-1)W_{12}\\
&\qquad\quad\quad\quad+(N-1)V_{12}\,\pmb\sigma_1\bullet\pmb\sigma_2,\widehat{m}\big](q_t)_1(q_t)_2\Psi_{\!N,t}\big\rangle_{\mathcal{H}_N}\,,
\end{split}
\end{equation}
and one has
\begin{equation} \label{eq:expansion_alpha}
\dot\alpha_{(\Psi_{\!N,t},\phi_t)}\;=\; 2(\mathrm{I})+2(\mathrm{II})+(\mathrm{III})+\mathrm{complex conjugate}\,.
\end{equation}

Let us now control separately the terms of each type in \eqref{eq:three_terms}. For the type $(\mathrm{I})$ we use the identities
\begin{equation} \label{eq:cancellation}
(p_t)_2W_{12}(p_t)_2\;=\;(p_t)_2 W_1^{\phi_t},\qquad \quad(p_t)_2V_{12}\,\pmb\sigma_1\bullet \pmb\sigma_2(p_t)_2\;=\;(p_t)_2 V_1^{\phi_t}
\end{equation}
that follow directly from the definition \eqref{eq:identities}, which allow us to write
\begin{equation*}
(\mathrm{I})\;=\;\frac{-\ii}{2}\big\langle\Psi_{\!N,t},(p_t)_1(p_t)_2\big[ W^{\phi_t}_1+V^{\phi_t}_1,\widehat{m}\big](q_t)_1(p_t)_2\Psi_{\!N,t}\big\rangle_{\mathcal{H}_N}\,.
\end{equation*}
We now use Lemma \ref{lemma:commutation} to compute the commutator, and consequently Lemma \ref{lemma:mn}(i) with $d=-1$, obtaining
\begin{equation*}
(\mathrm{I})\;=\;\frac{-\ii}{2N}\big\langle\Psi_{\!N,t},(p_t)_1(p_t)_2\big( W^{\phi_t}_1+V^{\phi_t}_1\big)(q_t)_1(p_t)_2\Psi_{\!N,t}\big\rangle_{\mathcal{H}_N}\,.
\end{equation*}
The smeared potentials inside brackets, as well the $p_t$'s and $q_t$'s, can be extracted in operator norm: Lemma \ref{lemma:smeared} then yields
\begin{equation} \label{eq:finalI}
|(\mathrm{I})|\;\lesssim\;\frac{1}{N}\,.
\end{equation}
Remarkably, the identities \eqref{eq:cancellation} allowed for the crucial removal of a $O(N)$ factor from $(\mathrm{I})$: this is a signature of the fact that the effective theory that we are considering is a good approximation of the complete theory.

Concerning the term of type $(\mathrm{II})$, using Lemma \ref{lemma:commutation} and then Lemma \ref{lemma:mn} with $d=-1$, we are able to expand the commutator and obtain 
\begin{equation*}
\begin{split}
(\mathrm{II})\;=\;\frac{\ii}{2}\big\langle \Psi_{\!N,t},(q_t)_1(p_t)_2&\times\big({\textstyle\frac{N-1}{N}}\,W_{12}+{\textstyle\frac{N-1}{N}}\,V_{12}\,\pmb{\sigma}_1\bullet\pmb\sigma_2-W_1^{\phi_t}-V_1^{\phi_t} \big)\\
&\times(q_t)_1(q_t)_2\Psi_{\!N,t}\big\rangle_{\mathcal{H}_N}\,.
\end{split}
\end{equation*}
Let us control the terms relative to the summands $\frac{N-1}{N}V_{12}\,\pmb{\sigma}_1\bullet\pmb\sigma_2$ and $V_1^{\phi_t}$, the corresponding terms with $W$ instead of $V$ being treated analogously. By means of Lemma \ref{lemma:smeared} and of \eqref{eq:identities} we get at once
\begin{equation*}
\big|\big\langle \Psi_{\!N,t},(q_t)_1(p_t)_2 V_1^{\phi_t}(q_t)_1(q_t)_2\Psi_{\!N,t}\big\rangle_{\mathcal{H}_N}\big| \;\lesssim\; \|(q_t)_1\Psi_{\!N,t}\|_{\mathcal{H}_N}^2\;=\;\alpha_{(\Psi_{\!N,t},\phi_t)}\,.
\end{equation*}
Moreover,
\begin{equation*}
\begin{split}
\big|\big\langle \Psi_{\!N,t},(q_t)_1&(p_t)_2 V_{12}\,\pmb\sigma_1\bullet\pmb\sigma_2(q_t)_1(q_t)_2\Psi_{\!N,t}\big\rangle_{\mathcal{H}_N}\big|\\
&\leqslant\; \sqrt{ \big\langle \Psi_{\!N,t},(q_t)_1(p_t)_2 V_{12}^2\big(\pmb\sigma_1\bullet\pmb\sigma_2\big)^2(p_t)_2(q_t)_1\Psi_{\!N,t}\big\rangle_{\mathcal{H}_N}}\;\times\\
&\qquad\quad \times\sqrt{ \big\langle \Psi_{\!N,t},(q_t)_1(q_t)_2\Psi_{\!N,t}\big\rangle_{\mathcal{H}_N} }\\
&\lesssim\;\big\|E_{V,\phi_t}\big\|_{\mathrm{op}}\|(q_t)_1\Psi_{\!N,t}\|_{\mathcal{H}_N}^2\;\lesssim\; \|(q_t)_1\Psi_{\!N,t}\|_{\mathcal{H}_N}^2\;=\;\alpha_{(\Psi_{\!N,t},\phi_t)}\,,
\end{split}
\end{equation*}
having applied the Cauchy-Schwarz inequality in the first step, the identity
\begin{equation}
(p_t)_2V^2_{12}(\pmb\sigma_1\bullet\pmb\sigma_2)^2(p_t)_2\;=\;(p_t)_2(E_{V,\phi_t})_1
\end{equation}
that follows from the definition \eqref{eq:smeared potentials} in the second step,  Lemma \ref{lemma:smeared} in the third step, and \eqref{eq:identities} in the last identity. Summarising, 
\begin{equation} \label{eq:finalII}
|(\mathrm{II})|\;\lesssim \;\alpha_{(\Psi_{\!N,t},\phi_t)}\,.
\end{equation}

Concerning the term of type $(\mathrm{III})$, using Lemma \ref{lemma:commutation} and then Lemma \ref{lemma:mn} with $d=-2$, we expand the commutator and obtain
\begin{equation*}
(\mathrm{III})\;=\;\ii\,{\textstyle\frac{N-1}{N}}\big\langle\Psi_{\!N,t},(p_t)_1(p_t)_2 \big(W_{12}+V_{12}\,\pmb\sigma_1\bullet\pmb\sigma_2\big)(q_t)_1(q_t)_2\Psi_{\!N,t}\big\rangle_{\mathcal{H}_N}\,.
\end{equation*}
Owing to Lemma \ref{lemma:mn}(ii) we can insert $\mathbbm{1}=\widehat{n}\widehat{n}^{-1}$ and get
\begin{equation*}
\begin{split}
(\mathrm{III})\;&=\;\ii\,{\textstyle\frac{N-1}{N}}\big\langle\Psi_{\!N,t},(p_t)_1(p_t)_2 \big(W_{12}+V_{12}\,\pmb\sigma_1\bullet\pmb\sigma_2\big)\widehat{n}\widehat{n}^{-1}(q_t)_1(q_t)_2\Psi_{\!N,t}\big\rangle_{\mathcal{H}_N}\\
&=\;\ii\,{\textstyle\frac{N-1}{N}}\big\langle\Psi_{\!N,t},(p_t)_1(p_t)_2\,\widehat{n_2} \big(W_{12}+V_{12}\,\pmb\sigma_1\bullet\pmb\sigma_2\big)\widehat{n}^{-1}(q_t)_1(q_t)_2\Psi_{\!N,t}\big\rangle_{\mathcal{H}_N},
\end{split}
\end{equation*}
having used Lemma \ref{lemma:commutation} to migrate $\widehat{n}$ to the left of the potentials. A Cauchy-Schwarz inequality yields
\begin{equation*}
\begin{split}
|(\mathrm{III})|\;&\leqslant\;\sqrt{\big\langle\Psi_{\!N,t},(p_t)_1(p_t)_2\,\widehat{n_2} \big(W_{12}+V_{12}\,\pmb\sigma_1\bullet\pmb\sigma_2\big)^2\widehat{n_2}(p_t)_1(p_t)_2\Psi_{\!N,t}\big\rangle_{\mathcal{H}_N}}\;\times \\
&\qquad \times\sqrt{\big\langle\Psi_{\!N,t},\widehat{n}^{-2}(q_t)_1(q_t)_2\Psi_{\!N,t}\big\rangle_{\mathcal{H}_N}}\,.
\end{split}
\end{equation*}
The term on the second line above is estimated by $\alpha_{(\Psi_{\!N,t},\phi_t)}^{1/2}$ by first recognising that $\widehat{n}^{-2}=\widehat{m}^{-1}$ and then using Lemma \ref{lemma:mn}(iii). For the term on the first line we observe that
\begin{equation}\label{eq:Fproperty}
(p_t)_2\big(W_{12}+V_{12}\,\pmb\sigma_1\bullet\pmb\sigma_2\big)^2(p_t)_2\;=\; (p_t)_2(F_{W,V,\phi_t})_1\,,
\end{equation}
which follows from the definition \eqref{eq:smeared potentials}. Lemma \ref{lemma:smeared} and \eqref{eq:Fproperty} then imply
\begin{equation*}
|(\mathrm{III})|\;\lesssim\;\sqrt{\|F_{W,V,\phi_t}\|_{\mathrm{op}}}\; \|\widehat{n_2}\Psi_{\!N,t}\|_{\mathcal{H}_N}\;\alpha_{(\Psi_{\!N,t},\phi_t)}^{1/2}\;\lesssim\;\|\widehat{n_2}\Psi_{\!N,t}\|_{\mathcal{H}_N}\;\alpha_{(\Psi_{\!N,t},\phi_t)}^{1/2}\,.
\end{equation*}
Finally, using $\widehat{n_2}^{\,2}=\widehat{m_2}=\widehat{m}+\frac{2}{N}\mathbbm{1}$, we arrive at
\begin{equation} \label{eq:finalIII}
|(\mathrm{III})|\;\lesssim\;\alpha_{(\Psi_{\!N,t},\phi_t)}+\frac{1}{N}\,.
\end{equation}

Using \eqref{eq:finalI}, \eqref{eq:finalII}, and \eqref{eq:finalIII} in \eqref{eq:expansion_alpha}, namely bounding the \emph{real} numbers (I)+complex conjugate, etc., we obtain
\begin{equation}\label{eq:alphadot_estimate}
\dot \alpha_{(\Psi_{\!N,t},\phi_t)}\;\leqslant\; C\Big( \alpha_{(\Psi_{\!N,t},\phi_t)}+\frac{1}{N}\Big)
\end{equation}
for a constant $C$ depending on $W,V,\phi_0$ but not on $N$ or $t$. Then, using a standard Gr\"onwall argument,
\begin{equation}
\alpha_{(\Psi_{\!N,t},\phi_t)}\;\leqslant\; e^{Ct}\Big(\alpha_{(\Psi_{\!N},\phi_0)}+\frac{1}{N}\Big)\;\leqslant\; \frac{K+1}{N}e^{Ct},
\end{equation}
having used assumption (iii) of Theorem \ref{thm:MFthm} in the second inequality in order to estimate the initial datum. This concludes the proof.

\section{Quantitative estimate of the fidelity of the model} \label{sect:phys}

In this final Section we present a quantitative analysis, based on recent experimental data, of the fidelity of the control \eqref{eq:alpha_at_t} of Theorem \ref{thm:MFthm} for the considered many-body mean-field model.

Our main conclusion here is going to be that, despite the character of `first approximation only' of the mean-field model, Theorem \ref{thm:MFthm} provides a control of the time-dependent indicator of condensation $\alpha_{(\Psi_{\!N_\mathrm{exp},t}^{\mathrm{phys}},\phi_t^{\mathrm{phys}})}$ that, \emph{for all the typical duration of an experiment on the dynamics of spinor condensates}, remains very small, of the order of the percent (or smaller). Thus, even the `rough' mean-field treatment given by Theorem \ref{thm:MFthm} provides a justification from first principles of the persistence of condensation throughout the observable dynamics of a spinor condensate.

%To do so, it is convenient to slightly change the assumptions of our analysis.

%Experimental values for the scattering lengths of the two channels $F=0$ and $F=2$ for the case of $^{87}\mathrm{Rb}$ are~\cite{Ho-1998} $a_0\simeq 55 \mathrm{\AA}$ and $a_2\simeq 54 \mathrm{\AA}$, whence $c_0\simeq 55 \mathrm{\AA}$ and $c_2\simeq 0.33 \mathrm{\AA}$. Recall that, according to our model, $c_2$ represents the scattering length of $V/N_{\mathrm{exp}}$ with $N_{\mathrm{exp}}$ the experimental number of particles. Since $c_2$ is much smaller than $c_0$, it is convenient to set $V=0$ for the purpose of the numerical estimate we are interested in.
%
%In order to model the physical value of $c_0$, in turn, we explicitly choose a soft sphere potential of the form

%depending on the two parameters $W_0$ and $R$ that need be tuned so as to have  that the scattering length of $W_{\mathrm{phys}}/N_{\mathrm{exp}}$ is precisely $c_0$. It is worth remarking that such choice allows for the very efficient estimate

%which is a consequence of Young inequality (an analogous estimate of course holds for $\|W_{\mathrm{phys}}^2*|\phi_t|^2\|_\infty$).

\subsection{Implementing experimental data}

In order to obtain quantitative estimates, let us revisit the setting and the proof presented in the previous Sections.

First we restore the physical constants $\hbar$ (Planck's constant) and $m$ (the mass of each boson atom), and we take $N=N_\mathrm{exp}$, the actual number of particles in a typical experiment. The many-body Hamiltonian \eqref{eq:meanfieldHN} takes the `physical' form
\begin{equation}\label{eq:HNphys}
\begin{split}
H_{N_{\mathrm{exp}}}^{\mathrm{phys}}\;&:=\;\sum_{j=1}^{{N_{\mathrm{exp}}}}(-\frac{\:\hbar^2}{2m}\Delta_{x_j})+\frac{1}{N_{\mathrm{exp}}}\sum_{1<j\leqslant k<N_{\mathrm{exp}}}W(x_j-x_k) \\
&\qquad\qquad+\frac{1}{N_{\mathrm{exp}}}\sum_{1<j\leqslant k<N_{\mathrm{exp}}}V(x_j-x_k)\,\pmb{\sigma}_j\bullet\,\pmb{\sigma}_k\,,
\end{split}
\end{equation}
the associated linear Schr\"{o}dinger equation for the `physical' many-body state $\Psi_{N,t}^{\mathrm{phys}}$ becomes 
\begin{equation*}
\ii\hbar\partial_t \Psi_{N,t}^{\mathrm{phys}}\;=\;H^{\mathrm{phys}}_{N_{\mathrm{exp}}} \Psi_{N,t}^{\mathrm{phys}}\,,
\end{equation*}
and the effective spinor Hartree system \eqref{eq:MF_system_compact} for the `physical' one-body orbital $\phi_t^{\mathrm{phys}}=\begin{pmatrix}
u_t^{\mathrm{phys}}\\ v_t^{\mathrm{phys}}\\ w_t^{\mathrm{phys}}
\end{pmatrix}$ takes the form
\begin{equation} \label{eq:MF_system_compact_phys}
\begin{split}
\ii\hbar\partial_t\phi_t^{\mathrm{phys}}\;&=\;-\frac{\:\hbar^2}{2m}\Delta\phi_t^{\mathrm{phys}}+W*\langle\phi_t^{\mathrm{phys}},\phi_t^{\mathrm{phys}}\rangle_{\mathbb{C}^3}\phi_t^{\mathrm{phys}}\\
& \qquad\qquad +V*\langle\phi_t^{\mathrm{phys}},\,\pmb{\sigma}\phi_t^{\mathrm{phys}}\rangle_{\mathbb{C}^3}\bullet\,\pmb{\sigma}\phi_t^{\mathrm{phys}}\,.
\end{split}
\end{equation}
Let us stress that \eqref{eq:HNphys} is \emph{not} a mean-field Hamiltonian: the factor $N_\mathrm{exp}^{-1}$ appears explicitly because of the present choice of the potentials $W$ and $V$, that are such that  $N_\mathrm{exp}^{-1}W$ and $N_\mathrm{exp}^{-1}V$ are the physical two-body potentials; then, when $N>N_\mathrm{exp}$, \eqref{eq:meanfieldHN} provides the mean-field re-scaled version of the physical Hamiltonian.

Carrying the physical constants over the various steps of Section \ref{sect:proof} for the proof of Theorem \ref{thm:MFthm} it is straightforward to see that formula \eqref{eq:expansion_alpha} now takes the form
\begin{equation}\label{eq:expansion_alpha_units}
\dot{ \alpha}_{(\Psi_{\!N_{\mathrm{exp}},t}^{\mathrm{phys}},\phi_t^{\mathrm{phys}})}\;\leqslant\;\hbar^{-1}\big( 2(\mathrm{I})+2(\mathrm{II})+\mathrm{III}+\mathrm{ complex}\,\,\mathrm{ conjugate}\big)\,,
\end{equation}
where $(\mathrm{I})$, $(\mathrm{II})$, and $(\mathrm{III})$ have the same expression as in \eqref{eq:three_terms} with $N=N_{\mathrm{exp}}$.
%and $\pmb{\sigma}$ replaced by $\hbar\,\pmb{\sigma}$. 
No explicit dependence on the mass $m$ appears, because of the cancellation of the kinetic terms that occurred in \eqref{eq:time_derivative_alpha}. Proceeding along the proof, one finds
\begin{equation} \label{eq:three_terms_phys}
\begin{split}
\big|(\mathrm{I})\big|\;&\leqslant\;\frac{1}{2N_{\mathrm{exp}}}\big(\, \big\|W^{\phi_t^{\mathrm{phys}}}\big\|_{\mathrm{op}}+\big\|V^{\phi_t^{\mathrm{phys}}}\big\|_{\mathrm{op}} \big)\\
\big|(\mathrm{II})\big|\;&\leqslant\;\frac{1}{2}\big(\,\big\|W^{\phi_t^{\mathrm{phys}}}\big\|_{\mathrm{op}}+\big\|V^{\phi_t^{\mathrm{phys}}}\big\|_{\mathrm{op}} +\sqrt{\big\|F_{W,V,\phi_t^{\mathrm{phys}}}\|_{\mathrm{op}}}\,\big)\,\alpha_{(\Psi_{\!N_{\mathrm{exp}},t}^{\mathrm{}},\phi_t^{\mathrm{phys}})}\\
\big|(\mathrm{III})\big|\;&\leqslant\; 2\sqrt{2}\,\sqrt{\big\|F_{W,V,\phi_t^{\mathrm{phys}}}\|_{\mathrm{op}}}\big( \alpha_{(\Psi_{\!N_{\mathrm{exp}},t}^{\mathrm{phys}},\phi_t^{\mathrm{phys}})}+N_{\mathrm{exp}}^{-1}\big)\,.
\end{split}
\end{equation}

For the Gr\"onwall estimate following from \eqref{eq:expansion_alpha_units}-\eqref{eq:three_terms_phys} to take an expression that is quantitatively informative, we need to qualify a physically realistic initial state $\Psi_{\!N_{\mathrm{exp}},0}^{\mathrm{phys}}$ and physically realistic potentials $W_{\mathrm{phys}}$ and $V_{\mathrm{phys}}$.

The many-body initial state must exhibit complete condensation $\Psi_{\!N_{\mathrm{exp}},0}^{\mathrm{phys}}\sim\big(\phi_0^{\mathrm{phys}}\big)^{\otimes N_{\mathrm{exp}}}$ in the quantitative sense \eqref{eq:alpha_at_tzero} for the reduced marginal. To be precise, since by construction $\alpha_{(\Psi_{\!N_{\mathrm{exp}},0}^{\mathrm{phys}},\phi_0^{\mathrm{phys}})}$ expresses the initial \emph{depletion} of the Bose gas (i.e., the fraction of particles that do not participate in the condensation), the constant $K N_\mathrm{exp}^{-1}$ in the bound \eqref{eq:alpha_at_tzero} at $t=0$ must bound from above the experimental value for the depletion. In order to match the typical values of depletion (Table \ref{tab:values}), we take
\begin{equation}\label{eq:alpha0exp}
\alpha_{(\Psi_{\!N_{\mathrm{exp}},0}^{\mathrm{phys}},\phi_0^{\mathrm{phys}})}\;\simeq\;4\cdot 10^{-3}\,.
\end{equation}
We shall qualify $\phi_0^{\mathrm{phys}}$ further in the following.
\renewcommand{\arraystretch}{1.6}
\begin{table}
	\begin{center}
		\begin{tabular}{|c|c|c|}
			\hline
			& \emph{experimental value} & \emph{source} \\
			\hline 
			\hline
			$^{87}\mathrm{Rb}$ atomic mass & $1.42\cdot 10^{-25}\; \mathrm{Kg}$ &\\
			\hline
			scattering lengths  & $\begin{array}{c}
			a_0=58.2 \mbox{{\AA}},\quad a_2=56.6\mbox{{\AA}} \\
			\Rightarrow \;\;c_0=57.1\mbox{{\AA}},\quad c_2=-0.53 \mbox{{\AA}}
			\end{array}$  & \cite{Ho-1998} \\
			\hline
			condensate population & $N_\mathrm{exp}=3\cdot 10^4\div 3\cdot 10^5$ & \cite{Chang-etAl-PRL2004spinor,Chang-Qin-Zhang-You-2005}\\
			\hline
			$\begin{array}{c}
			\textrm{condensate density }(n) \\
			\textrm{and depletion } (\alpha_0)
			\end{array}$ & $\begin{array}{c} n=10^{20} \:\mathrm{m}^{-3} \\ \Rightarrow\;\alpha_0=4\cdot 10^{-3}\end{array}$ & \cite{Chang-etAl-PRL2004spinor,Chang-Qin-Zhang-You-2005,Lopes-et-al-PRL2017} \\
			\hline
			condensate size & $R=10^{-4}\:\mathrm{m}$ & \cite{Chang-etAl-PRL2004spinor,Chang-Qin-Zhang-You-2005} \\
			\hline
			equilibration time & $T\lesssim 0.6\:\mathrm{sec}$  & \cite{Ketterle_StamperKurn_SpinorBEC_LesHouches2001,Chang-etAl-PRL2004spinor,Chang-Qin-Zhang-You-2005} \\
			\hline
		\end{tabular} \vspace{0.2cm}
		\medskip
		\caption{Experimental values of relevant quantities in typical modern experiments with the dynamical evolution of spinor condensates.}\label{tab:values}
	\end{center}
\end{table}
\renewcommand{\arraystretch}{1}

Concerning the potentials $W_{\mathrm{phys}}$ and $V_{\mathrm{phys}}$, for a first rough estimate it is enough to only consider the former and neglect the latter, since in a typical spinor condensate with $^{87}$Rb atoms the scattering length $c_2$ of $V_{\mathrm{phys}}$ is by far dominated by the scattering length $c_0$ of $W_{\mathrm{phys}}$ (Table \ref{tab:values}). In a crude approximation we model $W_{\mathrm{phys}}$ as the soft-sphere potential
\begin{equation}\label{eq:softsphereW}
W_{\mathrm{phys}}(x)\;:=\;\begin{cases}
W_0 & |x|\; < \; R\\
0 &|x|\,\geqslant\;R\,,\\\end{cases}
\end{equation}
with a radius $R$ that we take to be of the order of the condensate size (Table \ref{tab:values}) and a magnitude $W_0$ fixed by the requirement for $W_{\mathrm{phys}}$ to have scattering length $c_0$. An exact calculation based on \eqref{eq:softsphereW} \cite[Eq.~(84.8)]{Fluegge-practicalQM} shows that
\begin{equation*}
c_0= R\Bigg[ 1-\frac{\tanh\Big(\sqrt{\frac{W_0 m}{N_{\mathrm{exp}}\hbar^2}}R\Big)}{\sqrt{\frac{W_0 m}{N_{\mathrm{exp}}\hbar^2}}R} \Bigg]\,,
\end{equation*}
from which one can compute $W_0$ given $c_0$. In fact, it is straightforward to check by plugging the experimental values in the above formula that already the first Born approximation of $c_0$, i.e.,
\begin{equation}
8\pi c_0 \;\simeq\; \frac{2 m}{\,\hbar^2N_{\mathrm{exp}}}\int_{\mathbb{R}^3}\!W_{\mathrm{phys}}(x)\,\ud x \;=\; \frac{8\pi m W_0R^3}{3 \hbar^2 N_{\mathrm{exp}}}\,,
\end{equation}
yields a result that is very close to the exact one; thus, we take
\begin{equation}\label{eq:W0approx}
W_0\;\simeq\;\frac{\,3\,c_0\,\hbar^2N_{\mathrm{exp}}}{mR^3}\,.
\end{equation}
With the data of Table \ref{tab:values}, taking for concreteness $N_\mathrm{exp}\simeq 10^5$, we find $W_0\simeq 1.34\cdot 10^{-34}\,\mathrm{J}$, and hence
\begin{equation}\label{eq:W0hbar}
W_0/\hbar\;\simeq\; 1.3\cdot\mathrm{sec}^{-1}\,.
\end{equation}

\subsection{Estimates on the $\alpha$-indicator
%$\alpha_{(\Psi_{\!N_{\mathrm{exp}},t}^{\mathrm{phys}},\phi_t^{\mathrm{phys}})}$
}

With these data at hand, we turn back to the Gr\"onwall estimate that follows from \eqref{eq:expansion_alpha_units}-\eqref{eq:three_terms_phys}. Since we are neglecting $V_{\mathrm{phys}}$, the operator $F_{W,V,\phi_t^{\mathrm{phys}}}$ takes precisely the expression of $D_{W,\phi_t^{\mathrm{phys}}}$ and we find 
\begin{equation} \label{eq:estimate_V=0}
\begin{split}
\dot \alpha_{(\Psi_{\!N_{\mathrm{exp}},t}^{\mathrm{phys}},\phi_t^{\mathrm{phys}})}\;&\leqslant\;\hbar^{-1}\big(  2\,\big\|W^{\phi_t^{\mathrm{phys}}}\big\|_{\mathrm{op}}+(4\sqrt2+2)\sqrt{\big\|D_{W_\mathrm{phys},\phi_t^{\mathrm{phys}}}\big\|_{\mathrm{op}}}\,\big)\;\times\\
&\qquad\times\big(\alpha_{(\Psi_{\!N_{\mathrm{exp}},t}^{\mathrm{phys}},\phi_t^{\mathrm{phys}})}+N_{\mathrm{exp}}^{-1}\big)\,.
\end{split}
\end{equation}
With our choice \eqref{eq:softsphereW} for $W_{\mathrm{phys}}$, a Young inequality in \eqref{eq:smeared potentials} yields 
\begin{equation}\label{eq:young_bounded}
\big\|W^{\phi_t^{\mathrm{phys}}}\big\|_{\mathrm{op}}\;\leqslant\; \|W_{\mathrm{phys}}\|_{L^\infty(\mathbb{R}^3)}\,\|\phi_t^{\mathrm{phys}}\|_{\mathfrak{h}}^2\;=\;W_0
\end{equation}
and similarly
\begin{equation}\label{eq:young_bounded-2}
\big\|D_{W_\mathrm{phys},\phi_t^{\mathrm{phys}}}\big\|_{\mathrm{op}}\;\leqslant\;W_0^2\,.
\end{equation}

Plugging \eqref{eq:young_bounded} and \eqref{eq:young_bounded-2} into \eqref{eq:estimate_V=0}, and applying a standard Gr\"onwall argument, we obtain
\begin{equation}\label{eq:alpha_phys1}
\alpha_{(\Psi_{\!N_\mathrm{exp},t}^{\mathrm{phys}},\phi_t^{\mathrm{phys}})}\;\leqslant\;\big( \alpha_{(\Psi_{\!N_\mathrm{exp},0}^{\mathrm{phys}},\phi_0^{\mathrm{phys}})}+N_\mathrm{exp}^{-1}\big)\cdot e^{10\,t\,W_0/\hbar}\,.
\end{equation}
With $N_\mathrm{exp}$ taken from Table \ref{tab:values}, the initial value given by \eqref{eq:alpha0exp}, and $W_0/\hbar$ estimated as in \eqref{eq:W0hbar}, we see that formula \eqref{eq:alpha_phys1} produces a control on the indicator of condensation that for times $t\simeq 100\,\mathrm{msec}$, namely of the same order of the duration time of the experiment, is as accurate as
\begin{equation} 
\alpha_{(\Psi_{\!N,t}^{\mathrm{phys}},\phi_t^{\mathrm{phys}})}\;\leqslant\;0.015\,,
\end{equation}
thus less than 2\%.

As crude as the above estimate is, it shows that the mean-field scaling produces quite an accurate control of the dynamical persistence of condensation, when specialised with the actual experimental values.

In the above computation, we turned estimate \eqref{eq:alphadot_estimate} into the quantitative form \eqref{eq:estimate_V=0} and then we quantified the operator norms in \eqref{eq:estimate_V=0} through the bounds \eqref{eq:young_bounded}-\eqref{eq:young_bounded-2}. However, this neither exploited the fast spatial decay (in fact, the  finiteness of the support) of the two-body potential \eqref{eq:softsphereW}, nor the relatively short duration of a typical experiment, and by means of such two features a somewhat more refined estimate is possible.

To this aim, we control the operator norms in \eqref{eq:estimate_V=0} through the $L^p$-norm of $W_\mathrm{phys}$, that is,
\begin{equation} \label{eq:norms_W}
\|W_\mathrm{phys}\|_{L^p(\mathbb{R}^3)}\;=\;W_0\Big(\frac{4}{3}\pi R^3\Big)^{1/p}\,.
%,\qquad \|W_\mathrm{phys}\|_{L^\infty(\mathbb{R}^3)}\;=\;W_0\,.
\end{equation}
Applying Young's inequality and the quantitative Sobolev inequality
\begin{equation}\label{eq:quantSob}
{\textstyle\frac{3}{4}(2\pi^2)^{\frac{2}{3}}}\|f\|_{L^{6}(\mathbb{R}^3)}^2\;\leqslant\;\|\nabla f\|_{L^2(\mathbb{R}^3)}^2
\end{equation}
to the definition  \eqref{eq:smeared potentials}, we find 
\begin{equation}\label{eq:newbounds}
\begin{split}
\big\|W^{\phi_t^{\mathrm{phys}}}\big\|_{\mathrm{op}}\;&\leqslant\;0.18\;\|W_\mathrm{phys}\|_{L^{\frac{3}{2}}(\mathbb{R}^3)}\,\|\nabla\phi_t^{\mathrm{phys}}\|_{\mathfrak{h}}^2\\
\big\|D_{W_\mathrm{phys},\phi_t^{\mathrm{phys}}}\big\|_{\mathrm{op}}\;&\leqslant\;0.18\;\|W_\mathrm{phys}\|_{L^3(\mathbb{R}^3)}^2\,\|\nabla\phi_t^{\mathrm{phys}}\|_{\mathfrak{h}}^2\,,
\end{split}
\end{equation}
%%%%%%%%%%%%%%% OLD VERSION
% \begin{equation}\label{eq:newbounds}
% \begin{split}
% \big\|W^{\phi_t^{\mathrm{phys}}}\big\|_{\mathrm{op}}\;=&\;\|W\|_{L^{\frac{3}{2}}(\mathbb{R}^3)}\,\|\phi_t^{\mathrm{phys}}\|_{L^6(\mathbb{R}^3)}^2\;\leqslant\;4\,\|W\|_{L^{\frac{3}{2}}(\mathbb{R}^3)}\,\|\nabla\phi_t^{\mathrm{phys}}\|_{L^2(\mathbb{R}^3)}^2\\
% \big\|D_{W_\mathrm{phys},\phi_t^{\mathrm{phys}}}\big\|_{\mathrm{op}}\;=&\;\|W^2\|_{L^{\frac{3}{2}}(\mathbb{R}^3)}\,\|\phi_t^{\mathrm{phys}}\|_{L^6(\mathbb{R}^3)}^2\;\leqslant\;4\,\|W\|_{L^3(\mathbb{R}^3)}^2\,\|\nabla\phi_t^{\mathrm{phys}}\|_{L^2(\mathbb{R}^3)}^2\,,
% }
% \end{equation}
and plugging \eqref{eq:newbounds} into \eqref{eq:estimate_V=0} now yields
\begin{equation}\label{eq:alpha_phys2_partial}
\begin{split}
&\dot \alpha_{(\Psi_{\!N_{\mathrm{exp}},t}^{\mathrm{phys}},\phi_t^{\mathrm{phys}})}\;\leqslant\;\hbar^{-1}\big( \alpha_{(\Psi_{\!N_{\mathrm{exp}},t}^{\mathrm{phys}},\phi_t^{\mathrm{phys}})}+N_\mathrm{exp}^{-1}\big)\;\times \\
&\qquad\times\big(0.37\,\|W_\mathrm{phys}\|_{L^{\frac{3}{2}}(\mathbb{R}^3)}\|\nabla\phi_t^{\mathrm{phys}}\|_{\mathfrak{h}}^2+3.24\:\|W_\mathrm{phys}\|_{L^3(\mathbb{R}^3)}\|\nabla\phi_t^{\mathrm{phys}}\|_{\mathfrak{h}}\big).\!\!
\end{split}
\end{equation}

With the value $R\simeq 10^{-4}\,\mathrm{m}$ (Table \ref{tab:values}) we compute from \eqref{eq:norms_W}
\begin{equation}\label{eq:WnormsNumeric}
\begin{split}
\|W_\mathrm{phys}\|_{L^{\frac{3}{2}}}\;&\simeq\; W_0\cdot 2.6\cdot 10^{-8}\;\mathrm{m}^2\\
\|W_\mathrm{phys}\|_{L^{3}}\;&\simeq\;W_0\cdot 1.6\cdot 10^{-4}\;\mathrm{m}\,.
\end{split}
\end{equation}
%(see also Remark \ref{rem:units} below).

Concerning the estimate of $\|\nabla\phi_t^{\mathrm{phys}}\|_{\mathfrak{h}}$, it is realistic to take the one-body orbital onto which the gas condensates at $t=0$ of the Gaussian form
\begin{equation}
\phi_0^{\mathrm{phys}}(x)\;=\;\frac{1}{\sqrt{3}\,(\pi\sigma^2)^{\frac{3}{4}}}\,e^{-|x|^2/(2\sigma^2)}\!\begin{pmatrix} 1 \\ 1 \\ 1 \end{pmatrix},
\end{equation}
as is the case when the condensate is prepared in a harmonic trap, with a width that we may take of the same order of the size of the condensate, say $\sigma=R\simeq 10^{-4}\,\mathrm{m}$. With this choice, 
\begin{equation} \label{eq:nabla_phi_0}
\|\nabla\phi_0^{\mathrm{phys}}\|_{\mathfrak{h}}\;=\;\sqrt{\frac{3}{2\sigma^2}}\;\simeq \;1.2\cdot 10^{4}\:\mathrm{m}^{-1}\,.
\end{equation}

Even though the quantity $\|\nabla\phi_t^{\mathrm{phys}}\|_\mathfrak{h}$ deteriorates in time, it clearly remains uniformly bounded, as a consequence of the conservation of the Hartree functional \eqref{eq:Hfunctional}, which takes here the form 
\[
\mathcal{E}_{\mathrm{phys}}^{\mathrm{H}}[\phi_t^{\mathrm{phys}}]\;\approx\;\frac{\:\hbar^2}{2m}\big\|\nabla\phi_t^{\mathrm{phys}}\big\|_{\mathfrak{h}}^2+\frac{1}{2}\langle\phi_t^{\mathrm{phys}},W_{\mathrm{phys}}^{\phi_t^{\mathrm{phys}}}\phi_t\rangle_\mathfrak{h}
\]
(having neglected $V_{\mathrm{phys}}$ in comparison to $W_\mathrm{phys}$,)
For instance, extracting $\|W^{\phi_t^{\mathrm{phys}}} \|_{\mathrm{op}}\leqslant W_0$ as in \eqref{eq:young_bounded}, one finds
\begin{equation*}
\|\nabla\phi_t^{\mathrm{phys}}\|_\mathfrak{h}^2\;\leqslant\;\|\nabla\phi_0^{\mathrm{phys}}\|_\mathfrak{h}^2+\frac{\,2mW_0}{\hbar^2}
\end{equation*}
uniformly in $t$. However, the above bound, that follows from the sole energy conservation, is still too crude: indeed, since $(\frac{\,2mW_0}{\hbar^2})^{\frac{1}{2}}\simeq 5.9\cdot 10^4\;\mathrm{m}^{-1}$, then 
\begin{equation}\label{eq:L2nablaphit}
\|\nabla\phi_t^{\mathrm{phys}}\|_{\mathfrak{h}}\;\lesssim \;6\cdot 10^{4}\:\mathrm{m}^{-1}\,,
\end{equation}
and plugging \eqref{eq:WnormsNumeric} and \eqref{eq:L2nablaphit} into \eqref{eq:alpha_phys2_partial}, and applying a standard Gr\"onwall argument, yields
\begin{equation}\label{eq:alpha_phys2}
\alpha_{(\Psi_{\!N_\mathrm{exp},t}^{\mathrm{phys}},\phi_t^{\mathrm{phys}})}\;\leqslant\;\big( \alpha_{(\Psi_{\!N_\mathrm{exp},0}^{\mathrm{phys}},\phi_0^{\mathrm{phys}})}+N_\mathrm{exp}^{-1}\big)\cdot e^{\,66\,t\,W_0/\hbar}\,,
\end{equation}
which is qualitatively of the same type of, but does not improve, the previous bound \eqref{eq:alpha_phys1}.

Instead, let us estimate $\|\nabla\phi_t^{\mathrm{phys}}\|_\mathfrak{h}$ by monitoring its time evolution. The integral (Duhamel) form for the spinor Hartree equation \eqref{eq:MF_system_compact_phys}, with $V_{\mathrm{phys}}$ neglected as compared to $W_\mathrm{phys}$, reads
\begin{equation}\label{eq:physDu}
\phi^{\mathrm{phys}}_t\;\approx\;e^{\ii\frac{\hbar t}{2m}\Delta}\phi^{\mathrm{phys}}_0-\frac{\ii}{\,\hbar}\int_0^t e^{\ii\frac{\hbar (t-s)}{2m}\Delta}\big((W_{\mathrm{phys}}*|\phi^{\mathrm{phys}}_s|^2)\phi^{\mathrm{phys}}_s\big)\,\ud s\,,
\end{equation}
whence
\begin{equation}\label{eq:DuNabla}
\|\nabla\phi_t^{\mathrm{phys}}\|_\mathfrak{h}\;\leqslant\; \|\nabla\phi_0^{\mathrm{phys}}\|_\mathfrak{h}+\frac{1}{\,\hbar}\int_0^t\big\|\nabla \big((W_{\mathrm{phys}}*|\phi^{\mathrm{phys}}_s|^2)\phi^{\mathrm{phys}}_s\big)\big\|_\mathfrak{h}\,\ud s\,.
\end{equation}
By means of H\"{o}lder's and Young's inequalities, and using $\|\phi_s^{\mathrm{phys}}\|_\mathfrak{h}=1$, it is straightforward to find
\begin{equation*}
\begin{split}
\big\|\nabla\big((&W_{\mathrm{phys}}*|\phi^{\mathrm{phys}}_s|^2)\phi^{\mathrm{phys}}_s\big)\big\|_\mathfrak{h} \\
&\leqslant\;\|(W_{\mathrm{phys}}*\nabla|\phi^{\mathrm{phys}}_s|^2)\phi^{\mathrm{phys}}_s\|_\mathfrak{h}+\|(W_{\mathrm{phys}}*|\phi^{\mathrm{phys}}_s|^2)\nabla\phi^{\mathrm{phys}}_s\|_\mathfrak{h} \\
&\leqslant\;2\|W_{\mathrm{phys}}\|_{L^{\infty}(\mathbb{R}^3)}\|\nabla\phi^{\mathrm{phys}}_s\|_\mathfrak{h}+\|W_{\mathrm{phys}}\|_{L^\infty(\mathbb{R}^3)}\|\nabla\phi^{\mathrm{phys}}_s\|_\mathfrak{h}\,,
\end{split}
\end{equation*}
that is, owing to \eqref{eq:norms_W},
\begin{equation}\label{eq:int_term}
\frac{1}{\,\hbar}\, \big\|\nabla\big((W_{\mathrm{phys}}*|\phi^{\mathrm{phys}}_s|^2)\phi^{\mathrm{phys}}_s\big)\big\|_\mathfrak{h}\;\leqslant\;3\,\frac{\,W_0}{\,\hbar}\|\nabla\phi^{\mathrm{phys}}_s\|_\mathfrak{h}\,.
\end{equation}
Now, plugging \eqref{eq:int_term} into \eqref{eq:DuNabla}, yields
\[
\sup_{t\in[0,T]}\|\nabla\phi_t^{\mathrm{phys}}\|_\mathfrak{h}\;\leqslant\;\|\nabla\phi_0^{\mathrm{phys}}\|_\mathfrak{h}+\frac{3W_0 T}{\hbar}\sup_{t\in[0,T]}\|\nabla\phi_t^{\mathrm{phys}}\|_\mathfrak{h}\,,
\]
where $T$ is the considered duration of the time evolution (see Table \ref{tab:values}), whence, for sufficiently small $T$,
\begin{equation}\label{eq:gen}
\|\nabla\phi_t^{\mathrm{phys}}\|_\mathfrak{h}\;\leqslant\;(1-3\hbar^{-1}W_0 T)^{-1}\|\nabla\phi_0^{\mathrm{phys}}\|_\mathfrak{h}\,,\qquad t\in[0,T]\,.
\end{equation}
Using the estimate \eqref{eq:nabla_phi_0}, recalling that $W/\hbar\simeq 1.3\cdot\mathrm{sec}^{-1}$, and taking $T\simeq 100\;\mathrm{msec}$, we finally find
\begin{equation}\label{eq:L2nablaphit-NEW}
\|\nabla\phi_t^{\mathrm{phys}}\|_\mathfrak{h}\;\leqslant\;1.97\cdot 10^{4}\:\mathrm{m}^{-1}\,,\qquad t\in[0,T]\,,
\end{equation}
which improves \eqref{eq:L2nablaphit}. Plugging \eqref{eq:WnormsNumeric} and \eqref{eq:L2nablaphit-NEW} into \eqref{eq:alpha_phys2_partial}, and applying a standard Gr\"onwall argument, yields
\begin{equation}\label{eq:alpha_phys3}
\alpha_{(\Psi_{\!N_\mathrm{exp},t}^{\mathrm{phys}},\phi_t^{\mathrm{phys}})}\;\leqslant\;\big( \alpha_{(\Psi_{\!N_\mathrm{exp},0}^{\mathrm{phys}},\phi_0^{\mathrm{phys}})}+N_\mathrm{exp}^{-1}\big)\cdot e^{\,14\,t\,W_0/\hbar}\,,\qquad t\in[0,T]\,,
\end{equation}
which is comparable with \eqref{eq:alpha_phys1}.

Moreover, keeping the $T$-dependent estimate \eqref{eq:gen} and reasoning as above, we obtain the control
\begin{equation}\label{eq:alpha_phys4}
\begin{split}
\alpha_{(\Psi_{\!N_\mathrm{exp},t}^{\mathrm{phys}},\phi_t^{\mathrm{phys}})}\;\leqslant&\;\;\big( \alpha_{(\Psi_{\!N_\mathrm{exp},0}^{\mathrm{phys}},\phi_0^{\mathrm{phys}})}+N_\mathrm{exp}^{-1}\big)\cdot e^{\,T \xi(T)\,W_0/\hbar}\,,\qquad t\in[0,T]\,, \\
\xi(T)\;:=&\;\;1.38\,(1-3.9\cdot\mathrm{sec}^{-1}\;T)^{-2}+6.22\,(1-3.9\cdot\mathrm{sec}^{-1}\;T)\,.
\end{split}
\end{equation}
Obviously, for $T\simeq 100\,\mathrm{msec}$ \eqref{eq:alpha_phys4} reproduces \eqref{eq:alpha_phys3} ($\xi(100\,\mathrm{msec})\simeq 14$). But if we consider shorter times -- the most significant and precise part of the experiment takes place within $T\simeq 50\,\mathrm{msec}$ or less \cite{Chang-Qin-Zhang-You-2005} -- then \eqref{eq:alpha_phys4} becomes more accurate than \eqref{eq:alpha_phys1}. 
%For example, up to times $T\simeq 10\,\mathrm{msec}$ \eqref{eq:alpha_phys1} yields $\alpha_{(\Psi_{\!N_\mathrm{exp},t}^{\mathrm{phys}},\phi_t^{\mathrm{phys}})}$

%{\textstyle\frac{1}{\hbar}}

% With $N_\mathrm{exp}$ taken from Table \ref{tab:values}, the initial value given by \eqref{eq:alpha0exp}, and $W_0/\hbar$ estimated as in \eqref{eq:W0hbar}, we see that formula \eqref{eq:alpha_phys2} produces a control on the indicator of condensation that for times $t\simeq 100\,\mathrm{msec}$, namely of the same order of the duration time of the experiment, is as accurate as
% \begin{equation} 
% \alpha_{(\Psi_{\!N,t}^{\mathrm{phys}},\phi_t^{\mathrm{phys}})}\;\leqslant\;6\cdot 10^{-3}\,.
% \end{equation}

% 
\begin{remark}[Remarks on the units]\label{rem:units}
	For the benefit or the reader let us highlight here a few comments concerning the dimensional computations of this Section, with the convention that $[\mathcal{Q}]$ denotes the units of the quantity $\mathcal{Q}$.

	For the one-body and the many-body wave functions one has obviously $[\phi^{\mathrm{phys}}]=(\mathrm{length})^{-\frac{3}{2}}$ and $[\Psi_N^{\mathrm{phys}}]=(\mathrm{length})^{-\frac{3N}{2}}$; this implies that the quantity $\alpha_{(\Psi_{\!N_\mathrm{exp},t}^{\mathrm{phys}},\phi_t^{\mathrm{phys}})}$ is adimensional.
	
	Convolutions introduce an additional (length)$^{3}$ to the dimension of the involved functions. Thus,
	\[
	\begin{split}
	\big[\,\big\|W_\mathrm{phys}*|\phi_t^{\mathrm{phys}}|^2\big\|_{L^\infty(\mathbb{R}^3)}\big]\;&=\;\big[W_\mathrm{phys}\big] \\
	\big[\,\big\|W^2_\mathrm{phys}*|\phi_t^{\mathrm{phys}}|^2\big\|_{L^\infty(\mathbb{R}^3)}\big]\;&=\;\big[W_\mathrm{phys}\big]^2\,,
	\end{split}
	\]
	which confirms that \eqref{eq:young_bounded} and \eqref{eq:young_bounded-2} are dimensionally correct, and
	\[
	\big[\textrm{r.h.s.~of }\eqref{eq:estimate_V=0}\big]\;=\;[\hbar]^{-1}\big[W_\mathrm{phys}\big]\;=\;(\mathrm{time})^{-1}\;=\;\big[\textrm{l.h.s.~of }\eqref{eq:estimate_V=0}\big]\,.
	\]
	
	Since
	\[
	\begin{split}
	[\|f\|^2_{L^6(\mathbb{R}^3)}]\;&=\;(\mathrm{length})^{-2}\;=\;[\|\nabla f\|^2_{L^2(\mathbb{R}^3)}] \\
	\big[\|W_\mathrm{phys}\|_{L^p(\mathbb{R}^3)}\big]\;&=\;\big[W_\mathrm{phys}\big]\cdot(\mathrm{length})^{\frac{3}{p}}\,,
	\end{split}
	\]
	then \eqref{eq:quantSob} too is dimensionally correct, both sides of \eqref{eq:newbounds} have the dimension of $W_\mathrm{phys}$, and 
	\[
	\big[\textrm{r.h.s.~of }\eqref{eq:alpha_phys2_partial}\big]\;=\;[\hbar]^{-1}\big[W_\mathrm{phys}\big]\;=\;(\mathrm{time})^{-1}\;=\;\big[\textrm{l.h.s.~of }\eqref{eq:alpha_phys2_partial}\big]\,.
	\]

	The integral term in \eqref{eq:physDu} has the units of
	\[
	[\hbar]^{-1}\cdot\big[W_\mathrm{phys}\big]\cdot[\phi_s^{\mathrm{phys}}]^3\cdot(\mathrm{length})^3\cdot[\ud s]\;=\;(\mathrm{length})^{-\frac{3}{2}}
	\]
	(the $(\mathrm{length})^3$-contribution coming from the convolution), and indeed the whole \eqref{eq:physDu} is a $(\mathrm{length})^{-\frac{3}{2}}$. Reasoning in an analogous manner, one sees that both sides in \eqref{eq:int_term} are a $(\mathrm{length})^{-1}(\mathrm{time})^{-1}$, and indeed the whole \eqref{eq:DuNabla} is a $(\mathrm{length})^{-1}$.
	%  
	%  Now, a triangular inequality in $\|\nabla((W_{\mathrm{phys}}*|\phi^{\mathrm{phys}}_s|^2)\phi^{\mathrm{phys}}_s)\|_\mathfrak{h}$ yielded two terms, the first of which, namely $2\|W_{\mathrm{phys}}*(|\phi_s^{\mathrm{phys}}|\,|\nabla\phi_s^{\mathrm{phys}}|)\|_{L^2(\mathbb{R}^3)}$
	%  
	%  , and so is each of the two terms arising from the first triangular inequality used in the control of $\|\nabla((W_{\mathrm{phys}}*|\phi^{\mathrm{phys}}_s|^2)\phi^{\mathrm{phys}}_s)\|_\mathfrak{h}$
\end{remark}

\subsection{Comparison theory/experiment for the $\alpha$-indicator}\label{Sec:comparison-alpha}

\begin{figure}
\begin{center}
\includegraphics[width=9cm]{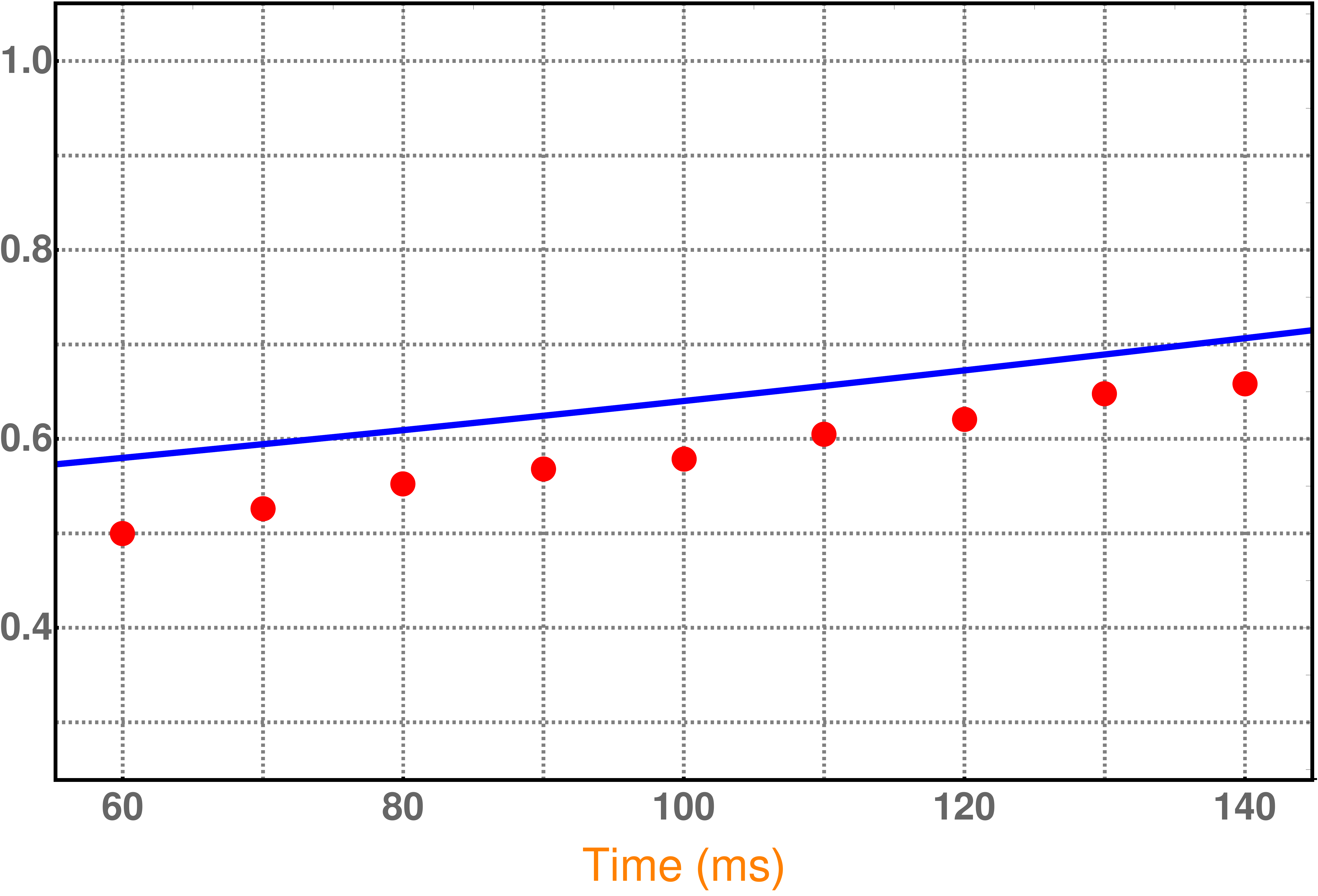}
\caption{Comparison for the behaviour in time of the $\alpha$-indicator of condensation. Blue solid curve: theoretical bound on $\alpha_{(\Psi_{\!N_\mathrm{exp},t}^{\mathrm{phys}},\phi_t^{\mathrm{phys}})}$ according to formula (\ref{eq:alpha_phys1}) with parameters inferred from the experiment \cite[Fig.~5(a)]{Chang-etAl-PRL2004spinor}. Red dotted curve: experimental values of $\alpha_{(\Psi_{\!N_\mathrm{exp},t}^{\mathrm{phys}},\phi_t^{\mathrm{phys}})}$ reconstructed from the data of the same experiment \cite[Fig.~5(a)]{Chang-etAl-PRL2004spinor}. Discussion in Subsection \ref{Sec:comparison-alpha}.} \label{fig:alpha-indicators}
\end{center}
\end{figure}

One last quantitative evidence of the good fidelity of the model with respect to the experiments is the following.

We extracted from the experiment \cite{Chang-etAl-PRL2004spinor} with spinor condensates (already considered for Table \ref{tab:values}) the value of the total number of particles of the spinor gas, $N_{\textrm{exp}}^{\textrm{(tot)}}\sim 1.9
\cdot 10^4$, and the initial number of particles $N_{\textrm{exp}}^{\textrm{(init)}}\sim 9.5 \cdot 10^3$ participating to the spinor condensate when the proper dynamical expansion of the condensate starts after some transient regime, namely after approximatively 60 msec in \cite[Fig.~5(a)]{Chang-etAl-PRL2004spinor}. This allows us to estimate the initial $\alpha$-indicator for Eq.~(\ref{eq:alpha_phys1}) as $\alpha_{(\Psi_{\!N_\mathrm{exp},0}^{\mathrm{phys}},\phi_0^{\mathrm{phys}})}\sim 1-N_{\textrm{exp}}^{\textrm{(init)}}/N_{\textrm{exp}}^{\textrm{(tot)}}\sim 0.50$ and also, through formula (\ref{eq:W0approx}), the updated value $W_0/\hbar\sim 0.247\cdot\mathrm{sec}^{-1}$ for this experiment (which remains comparable with (\ref{eq:W0hbar})). Plugging these parameters into (\ref{eq:alpha_phys1}) we can plot the time evolution of (an upper bound on) the $\alpha$-indicator as \emph{predicted theoretically} by our mathematical scheme for the considered experiment. This is the blue solid curve in Figure \ref{fig:alpha-indicators}.

In comparison to that, we also extracted from the same experiment \cite[Fig.~5(a)]{Chang-etAl-PRL2004spinor} the data on the number of particles $N_{\mathrm{exp}}(t)$ participating in the spinor condensate at later times in the course of the subsequent 80 msec of the dynamical evolution, and estimated correspondingly the \emph{experimental} behaviour of the $\alpha$-indicator as $1-N_{\mathrm{exp}}(t)/N_{\textrm{exp}}^{\textrm{(tot)}}$. These are the red dots in Figure \ref{fig:alpha-indicators}.

We can thus see that the theoretical prediction on the $\alpha$-indicator stays above and satisfactorily close to the experimental values of the same indicator \emph{for all the duration of the considered experiment}. The displacement between the two, besides all the reasonable approximations made so far, can be plausibly explained also in terms of temperature effects: our theoretical analysis is indeed carried on in a zero-temperature formalism, whereas the experimental data are taken at (ultra-cold, yet) finite temperature. Yet the agreement is appreciable and it confirms the good fidelity of the theoretical model for the derivation of the mean-field effective dynamical equations.

%  \bibliographystyle{siam}
%  \bibliography{bib_ALE}

\def\cprime{$'$}

\end{document}